%% file: Higher_Order_Logic_Locking.tex
\documentclass[runningheads]{llncs}
\usepackage{amsmath}
\usepackage{wrapfig}
\usepackage{graphicx}
\usepackage{breqn}
\usepackage[numbers]{natbib}
\usepackage{wrapfig}
\usepackage{multirow}
\usepackage{array}
\usepackage{listings}
\usepackage{caption}

\PassOptionsToPackage{nodisplayskipstretch}{setspace}
\usepackage[nodisplayskipstretch]{setspace}

\usepackage{colortbl}

\usepackage[utf8]{inputenc}
\usepackage[linewidth=1pt]{mdframed}
\usepackage{subcaption}
\usepackage{stmaryrd}
\usepackage{xcolor}
\usepackage{todonotes}
\usepackage[ruled,vlined,linesnumbered]{algorithm2e}

\newcolumntype{P}[1]{>{\centering\arraybackslash}p{#1}}
\newcommand{\toolname}{\textsc{HOLL}}

\newcommand{\lckd}{\hat{\varphi}}
\newcommand{\ckt}{\varphi}
\newcommand{\key}{\psi}
\newcommand{\bimp}{\leftrightarrow}

\newcommand{\figref}[1]{Figure~\ref{fig:#1}}
\newcommand{\tabref}[1]{Table~\ref{tab:#1}}
\newcommand{\algref}[1]{Algorithm~\ref{alg:#1}}

\newcolumntype{L}[1]{>{\raggedright\let\newline\\\arraybackslash\hspace{0pt}}m{#1}}
\newcolumntype{C}[1]{>{\centering\let\newline\\\arraybackslash\hspace{0pt}}m{#1}}
\newcolumntype{R}[1]{>{\raggedleft\let\newline\\\arraybackslash\hspace{0pt}}m{#1}}

\newcolumntype{X}{>{\centering\arraybackslash}m{1.1cm}}

\newcommand{\sem}[1]{\llbracket #1 \rrbracket}

\newcommand{\ocl}{\sem{\lckd(\key)}}

\newcommand{\B}[1]{\langle #1\rangle}

\begin{document}
\date{}
\title{\Large \bf HOLL: Program Synthesis for Higher Order Logic Locking}
\titlerunning{Program Synthesis for Higher Order Logic Locking}
%
\author{Gourav Takhar\inst{1} \and Ramesh Karri\inst{2} \and Christian Pilato\inst{3} \and Subhajit Roy\inst{1}}
\institute{Indian Institute of Technology Kanpur, UP, India
\and
New York University, New York, NY, USA
\and
Politecnico di Milano, Milan, Italy}
%
%

\maketitle
\begin{abstract}
  \input{content/abstract}
\end{abstract}
%

%
%
%

\input{content/introduction}
\input{content/overview}

\input{content/algorithm}
\input{content/attack}

\input{content/experimentation}

\input{content/related_work}
\input{content/discussion}

\subsubsection{Acknowledgements.}
We thank the anonymous reviewers for their valuable inputs. The second author is supported in part by ONR grant number N000141812058.
\bibliographystyle{splncs04}
\bibliography{biblio}


\end{document}

%% file: content/abstract.tex
Logic locking ``hides" the functionality of a digital circuit to protect it from counterfeiting, piracy, and malicious design modifications. The original design is transformed into a ``locked" design such that the circuit reveals its correct functionality only when it is ``unlocked" with a \textit{secret} sequence of bits---the \textit{key} bit-string. However, strong attacks, especially the \textit{SAT attack} that uses a SAT solver to recover the key bit-string, have been profoundly effective at breaking the locked circuit and recovering the circuit functionality. 

We lift logic locking to \textit{Higher Order Logic Locking (HOLL)} by hiding a higher-order \textit{relation}, instead of a key of independent values, challenging the attacker to discover this \textit{key relation} to recreate the circuit functionality. Our technique uses \textit{program synthesis} to construct the locked design and synthesize a corresponding \textit{key relation}. HOLL has low overhead and existing attacks for logic locking do not apply as the entity to be recovered is no more a value. To evaluate our proposal, we propose a new attack (\textit{SynthAttack}) that uses an inductive synthesis algorithm guided by an operational circuit as an input-output oracle to recover the hidden functionality. SynthAttack is inspired by the SAT attack, and similar to the SAT attack, it is \textit{verifiably correct}, i.e., if the correct functionality is revealed, a verification check guarantees the same. Our empirical analysis shows that SynthAttack can break HOLL for small circuits and small key relations, but it is ineffective for real-life designs.

%% file: content/introduction.tex
\section{Introduction}

High manufacturing costs in advanced technology nodes are pushing many semiconductor design houses to outsource the fabrication of integrated circuits (IC) to third-party foundries~\cite{fabless,iccost}. A {\em fab-less} design house can increase the investments in the chip's intellectual property, while a single foundry can serve multiple companies. However, this globalization process introduces security threats in the supply chain~\cite{pieee14}. A malicious employee of the foundry can access and reverse engineer the circuit design to make illegal copies. We can classify the protection methods in {\em passive} and {\em active}. Passive methods, like watermarking, allow designers to identify but not prevent illegal copies~\cite{abdel2003ip}. Active methods, like logic locking, alter the chip's functionality to make it unusable by the foundry~\cite{10.1145/3342099}. This alteration depends on a \textit{locking key} that is re-inserted into the chip in a trusted facility, after fabrication. The locking key is thus the ``secret'', known to the design house but unknown to the foundry.
Logic locking assumes that the attackers have no access to the key but they may have access to a functioning chip (obtained, for example, from the legal/illegal market).
However, logic locking has witnessed several attacks that analyze the circuit and attempt key recovery~\cite{8714955,Shamsi2017,sisejkovic2020challenging,removalAttack}.
\newline \indent
In this paper, we combine the intuitions from logic locking, program synthesis, and programmable devices to design a new locking mechanism. Our technique, called \textit{higher order logic locking} (HOLL), locks a design with a \textit{key relation} instead of a sequence of independent \textit{key bits}. HOLL uses \textit{program synthesis}~\cite{alur2015,solar-lezama2013} to translate the original design into a locked one.
Our experiments demonstrate that HOLL is fast, scalable, and robust against attacks. Prior attacks on logic locking, like the SAT attack~\cite{7140252}, are not practical for HOLL. Since the functionality of the key-relation is completely missing, attackers cannot simply make propositional logic queries to recover the key (like~\cite{Shamsi2017,7140252,massad15}). There are variants of logic locking, like TTLock~\cite{yasin2017} and SFLL~\cite{Yasin:2017:PLL:3133956.3133985}, that attempt to combat SAT attacks~\cite{7140252}. However, these techniques use additional logic blocks (comparison and restoration circuits) which makes them prone to attacks via structural and functional analysis on this additional circuitry~\cite{Deepak19}. HOLL is resilient to such techniques as it exposes only a programmable logic that does not leak any information related to the actual functionality to be implemented.
\newline \indent
In contrast to logic locking, attacking HOLL requires solving a second-order problem (see \S\ref{sec:Discussion} for a detailed discussion on this).
To assess the security of our method, we design a new attack, \textit{SynthAttack}, by combining ideas from SAT attack~\cite{7140252} and inductive program synthesis~\cite{solar-lezama2013}. SynthAttack employs a counter-example guided inductive synthesis (CEGIS) procedure guided via a functioning instance of the circuit as an input-output oracle. This attack constructs a synthesis query to discover key relations that invoke \textit{semantically distinct functional behaviors} of the locked design---witnesses to such relations, referred to as \textit{distinguishing inputs}, act as counterexamples to drive inductive learning. When the locked design is verified to have unique functionality, the attack is declared successful, with the corresponding provably-correct key relation.
\newline \indent
Our experimental results (\S\ref{sec:experimentaleval}) show that the time required by an attacker to recover the key relation for a given set of distinguishing inputs (\textit{attack time}) increases exponentially with the size of key relation. While the attacker may be able to recover key relations for small HOLL-locked circuits with small key relations, larger circuits are robust to SynthAttack.
For example, for the {\tt des} benchmark with \texttt{4,174} gates, the asymmetry between HOLL defense and SynthAttack is large; while HOLL can lock this design in less than \texttt{100} seconds, the attack cannot recover the design even within four days for a key relation that increases the area overhead of the IC by only \texttt{1.2}\%. Further, the attack time required to unlock the designs increase exponentially with the complexity of the key relation and the number of distinguishing inputs.
\newline \indent
The key relation can be implemented with reconfigurable or programmable devices, like programmable array logic (PAL) or embedded finite-programmable gate array (eFPGA).
For example, eFPGA, essentially an IP core embedded into an ASIC or SoC, is becoming common in modern SoCs~\cite{efpgaMarket} and has been shown to have high resilience against bit-stream recovery~\cite{iccad21}.
\newline \indent
Our contributions are:
\begin{itemize}
\setlength{\itemsep}{-0.5pt}
    \item We propose a novel IP protection strategy, called \textit{higher order logic locking} (HOLL), that uses program synthesis, challenging attackers to synthesize a key relation (as opposed to a key bit-string, as in logic locking);
    \item To evaluate the security offered by HOLL, we propose a strong adversarial attack algorithm, SynthAttack;
    \item We build tools to apply HOLL and SynthAttack to combinational logic;
    \item We evaluate HOLL on cost, scalability, and robustness.
\end{itemize}

%% file: content/overview.tex
\section{HOLL Overview} \label{sec:overview}

\input{content/background}

\subsection{Defending with \textit{HOLL}}
        Consider a hardware circuit $Y \bimp \ckt(X)$, where $X$ and $Y$ are the inputs and outputs, respectively. HOLL uses a higher-order lock---a \textit{secret} relation~($\key$) among a certain number of additional \textit{relation} bits $R$. We refer to $\key$ as the {\bf key relation}.

      \begin{figure}[t]
      \begin{subfigure}[b]{0.35\textwidth}
$t_0 = x_0 \land x_2$;\\
$t_1 = $ \fcolorbox{black}{blue!20}{$ x_3 \land t_0$}\\
$t_2 = $ \fcolorbox{black}{blue!20}{($x_1 \land t_0$)}\\
$y_0 = x_0 \oplus x_2$\\
$y_2 = (x_1 \land x_3) \lor t_2 \lor t_1$\\
$y_1 = t_0 \oplus x_1 \oplus x_3$
        \caption{Original circuit}
        \label{fig:adder}
      \end{subfigure}
        \begin{subfigure}[b]{0.39\textwidth}
$t_0 = x_0 \land x_2$\\
$t_1 = $\fcolorbox{black}{red!20}{$(x_0 \land ( r_4 \oplus r_2) \land x_3)$}\\
$t_2 = $\fcolorbox{black}{red!20}{$(x_0 \land r_3)$}\\
$y_0 = x_0 \oplus x_2$\\
$y_2 = (x_1 \land x_3) \lor t_2 \lor t_1$\\
$y_1 = t_0 \oplus x_1 \oplus x_3$
          \caption{Locked circuit}
          \label{fig:lockadder}
        \end{subfigure}
        \begin{subfigure}[b]{0.2\textwidth}
\{$(r_0 \bimp x_1)$,\\
$(r_1 \bimp x_2)$,\\
$(r_2 \bimp rand)$,\\
$(r_3 \bimp r_0 \land r_1)$,\\
$(r_4 \bimp r_1 \oplus r_2)$\}
            \caption{Key relation} \label{fig:key}
        \end{subfigure}
        \caption{HOLL on a 2-bit Adder.}
        \label{fig:adder_lock_}
      \end{figure}

        \figref{adder} shows an example of a \texttt{2}-bit adder with input $X$ (\{$x_1x_0$, $x_3x_2$\}) and output $Y$ ($y_2y_1y_0$). The circuit is locked by transforming the original expressions (marked in \textcolor{blue}{blue}) in \figref{adder} to the \textit{locked} expressions (marked in \textcolor{red}{red}) in \figref{lockadder}. The locked expressions use the additional \textit{relation} bits $r_2$, $r_3$, and $r_4$, enabling that this \textit{locked design} $\lckd(X,R)$ functions correctly when the secret relation $\key$ (\figref{key}) is installed. The relation $\key$ establishes the correct relation between the {\bf relation bits} $R$. The key relation can be excited by circuit inputs (like in $r_0$ and $r_1,$), constants, or random bits (e.g., from system noise, etc.); for example, the value $rand$ in \figref{key} represents the random generation of a bit (\texttt{0} or \texttt{1}) assigned to $r_2$. The ``output" from the key relation are bits $r_3$ and $r_4$ that must satisfy the relational constraints enforced by the key relation.

        For the sake of simplicity, in the rest of the paper, we assume the relation bits are drawn only from the inputs $X$ of the design. We will attempt to infer key relations of the form $\key(X, R)$. The reader may assume the value $rand$ of in \figref{key} to be a constant value (say $\texttt{0}$) to ease the exposition.

       As $\lckd$ also consumes the \textit{relation bits} $R$, HOLL transforms the original circuit $Y\bimp \ckt(X)$ into a \textit{locked} circuit $Y\bimp{\lckd}(X,R)$ such that the locked circuit functions correctly if the relational constraint $\key(X,R)$ is satisfied. In other words, HOLL is required to preserve the semantic equivalence between the original and locked designs ($\ckt = \lckd \land \key$). Note that it only imposes constraints on the input-output functionality of the circuits and the values generated off the internal gates may diverge. For example, in \figref{lockadder}, the value of $t_1$ may be different from the one in the original design (\figref{adder}), but the final output $y_2$ remains equivalent to the original adder.

        This approach has analogies with the well-known \textit{logic locking} solution~\cite{5401214,dac18,llcarxiv}. Traditional logic locking produces a locked circuit by mutating certain expressions based on input key bits. HOLL differs from logic locking on the type of entities employed as hidden keys. While logic locking uses a \textit{key value} (i.e., a sequence of key-bits), our technique uses a \textit{key relation} (i.e., a functional relation among the key bits). We refer to our scheme as \textit{higher-order logic locking} (HOLL). As synthesizing a relation (a second-order problem) is more challenging to recover than a bit-sequence (a first-order problem), HOLL is, at least in theory, is more secure than logic locking. Our experimental results (\S\ref{sec:experimentaleval}) show that this security also translates to practice.

\subsubsection{Hardware constraints.} Since the key relation must be implemented in the final circuit, we need to consider practical hardware constraints; details of the hardware implementation is provided in \S\ref{sec:hardware_architecture}.
For example, the size of the key relation affects the size of the programmable logic to be used for its implementation. This, in turn, introduces area and delay overheads in the final circuit.
The practical realizability of this technique as a hardware device adds certain constraints:

        \begin{itemize}
        \setlength{\itemsep}{0pt}
            \item The key relation must be small for it to have a small area overhead;
            \item The key relation must only be executed once to avoid a significant performance overhead;
            \item The key relation must encode non-trivial relations between the challenge and response bits to strong security;
            \item The locked expressions are evenly distributed across the design to protect all parts of the circuit, disallowing focused attacks by an attacker on a small part of the circuit that contains all locks.
         \end{itemize}
 \subsubsection{Inferring the key relation.}

        HOLL operates by
        \begin{enumerate}
        \setlength{\itemsep}{-0.5pt}
            \item  carefully selecting a set of expressions, $E \subseteq \ckt$, in the original design $\ckt$;
            \item mutating each expression $e_i \in E$ using the relation bits $R$ to create the corresponding \textit{locked expression}, $\hat{e}_i$.
        \end{enumerate}
        For example, in Figure \ref{fig:adder}, we select two expressions, $E=\{e_1,e_2\}$ where $e_1= x_1 \land t_0$ and $e_2= x_3 \land t_0$. $e_1$ computes $t_2$ and is a function of $t_0$ and $x_1$, while $\hat{e}_1$ uses $x_0$ and $r_3$, which is in turn a relation of $r_0$ and $r_1$.
We formalize our lock program synthesis problem as follows.
\paragraph{Lock Inference. } Given a circuit $Y\bimp \ckt(X)$, construct a \textit{locked circuit} $Y\bimp \lckd(X,R)$ and a \textit{key relation} $\key(X,R)$ such that 
$\lckd$ is \textit{semantically equivalent} to $\ckt$ with the correct relation $\key$. Specifically, it requires us to construct: (1) a key relation $\key$ and (2) a set of locked expressions $\hat{E}$ relating to the set of selected expressions $E$ extracted from $\ckt$ such that the following conditions are met:
\begin{itemize}
\setlength{\itemsep}{0pt}
    \item {\bf Correctness}: The circuit is guaranteed to work correctly for all inputs when the key relation is installed:
        \[ \forall X.\ (\forall R .\  {\key(X,R)\implies (\lckd}(X,R) =  \ckt(X)))\tag{1}\label{eqn:correctness}\]
        where $\lckd \equiv \ckt[\hat{e}_1/e_1, \dots, \hat{e}_n/e_n]$, for $e_i \in E\subseteq \ckt$. The notation $\ckt[e_a/e_b]$ implies that $e_b$ is replaced by $e_a$ in the formula $\ckt$.
    \item {\bf Security}: There must exist some relation $\key'$ (where $\key' \neq \key$) where the circuit exhibits incorrect behavior; in other words, it enforces the key relation to be non-trivial:
        \[ \exists \key\ \exists X\ \exists R.\ (\key'(X,R) \implies {\lckd}(X,R) \neq \ckt(X))\tag{2}\label{eqn:security}\]
\end{itemize}
We pose the above as a \textit{program synthesis}~\cite{bugsynthesis,sketch} problem. In particular, we search for ``mutations'' $\hat{e}_1, \dots, \hat{e}_2$ and a suitable key relation $\key$ such that the above constraints are satisfied.

\subsection{Attacking with \textit{SynthAttack}} As we attempt to hide a relation instead of a key-value, prior attacks on logic locking (like SAT attacks), which attempt to infer key bit-strings, do not apply. However, the attackers can also use program synthesis techniques to recover the key relation using an activated instance of the circuit as an input-output oracle.

We design an attack algorithm, called \textit{SynthAttack}, combining ideas from SAT attack (for logic locking) and counterexample guided inductive program synthesis. Our attack algorithm generates inputs $X_1, X_2, \dots, X_n$ and computes the corresponding outputs $Y_1, Y_2, \dots, Y_n$ using the oracle, to construct a set of \textit{examples} $\Lambda=\{(X_1, Y_1), \dots, (X_n, Y_n)\}$. Then, the attacker can generate a key relation $\key$ that satisfies the above examples, $\lambda$, using a program synthesis query:
        \[\prod_{X_i, Y_i\in \Lambda} \exists R_i.\  \lckd(X_i,R_i) \land \key(X_i, R_i) = Y_i\tag{3} \]
The above query requires \textit{copies} of ${\hat{\varphi}(X, R)}$ for every example---hence, the formula will quickly explode with an increasing number of samples. Our scheme is robust since the sample complexity of the key relationships increases exponentially with the number of relation bits employed. Additionally, the attacker does not know which input bits excite the key relation and how the relation bits are related to each other.

\begin{figure}[t]
\begin{minipage}{0.4\textwidth}
\centering
\captionof{table}{In-out samples.}\label{tab:inputsamples}
\begin{tabular}{ | c | c | c | }
\hline
\cellcolor{gray!40}$X$ &\cellcolor{gray!40}$Y$ &\cellcolor{gray!40}$\hat{Y}$ \\\hline\hline
\cellcolor{gray!0}1111 &\cellcolor{green!40}110 &\cellcolor{green!40}110  \\\hline
\cellcolor{gray!0}1001 &\cellcolor{green!40}011 &\cellcolor{green!40}011 \\\hline
\cellcolor{gray!0}0000 &\cellcolor{green!40}000 &\cellcolor{green!40}000\\\hline
\cellcolor{gray!0}1100 &\cellcolor{green!40}011 &\cellcolor{green!40}011\\\hline
\cellcolor{gray!0}0101 &\cellcolor{green!40}010 &\cellcolor{red!40}110\\\hline
\end{tabular}
\end{minipage}
\hfill
\begin{minipage}{0.4\textwidth}
\{$(r_0 \bimp x_2)$, \\
$(r_1 \bimp x_0)$, \\
$(r_2 \bimp 0)$,\\
$(r_3 \bimp r_0 \land r_1)$,\\
$(r_4 \bimp 0)$\}
  \captionof{figure}{Generated key relation}
  \label{fig:obtainedkey}
\end{minipage}
\end{figure}

For the locked adder (\figref{lockadder}) with the input samples shown in \tabref{inputsamples} (first four rows), the above attack can synthesize the key relation shown in \figref{obtainedkey}. Columns $Y$ and $\hat{Y}$ in \tabref{inputsamples} represent the outputs of the original circuit and the circuit obtained by the attacker, respectively. Even on a \texttt{4}-bit input space, when \texttt{25}\% of all possible samples are provided, the attack fails to recover the key relation as shown by the last input row of \tabref{inputsamples}. The red box highlights the output in the attacker circuit does not match the original design.
\paragraph{Definition (Distinguishing Input).} Given a locked circuit $\lckd$, we refer to input $X$ as a distinguishing input if there exist candidate relations $\key_1$ and $\key_2$ that evoke semantically distinct behaviors on the input $X$. Formally, $X$ is a distinguishing input provided the following formula is satisfiable\footnote{All free variables are existentially quantified} on some relations $\key_1$ and $\key_2$:
\begin{align}
\lckd(X, R_1) \neq \lckd(X, R_2) \land \key_1(X, R_1) \land \key_2(X, R_2)\tag{4}\label{eqn:distingInput}
\end{align}
It searches for a \textit{distinguishing input}, $X_d$, that produces conflicting outputs on the locked design. Any such distinguishing input is added to the set of examples, $\Lambda$, and the query repeated. If the query is unsatisfiable, it implies that no other relation can produce a different behavior on the locked design and so the relation that satisfies the current set of examples must be a correct key relation.

Though SynthAttack significantly reduces the sample complexity of the attack, our experiments demonstrate that SynthAttack is still unsuccessful at breaking HOLL for larger designs.

%% file: content/background.tex
\subsection{Threat Model: the Untrusted Foundry}\label{sec:threatmodel}
We focus on the threat model where the attacker is in the foundry~\cite{10.1145/3342099,8741035} to which a fab-less design house has outsourced its IC fabrication. Such an attacker has access to the IC design and the (locked) GDSII netlist which can be reverse-engineered.
Also, if the attacker can access a working IC (e.g., by procuring an IC from the open market or a discarded IC from the gray market), they can leverage the functional IC's I/O behavior as a black-box oracle.
However, we assume the attacker \textit{cannot} extract the {\em bitstream}, i.e. the correct sequence of configuration bits, from the device. This can be achieved with encryption techniques when the bitstream is not loaded into the device. Also, anti-readback solutions can prevent the attacker from reading the bitstream from the device. The parameters used to synthesize the \textit{key relation} and the \textit{locked circuit} (like the domain-specific language (DSL), budget etc.) are not known to the attacker (see \S\ref{sec:Discussion}).

%% file: content/algorithm.tex
\section{Program Synthesis to Infer Key Relations}\label{sec:progsynth}

We represent the key relation $\key$ as a propositional formula, represented as a set (conjunction) of propositional terms. The terms in $\key$ belong are categorised as:
\begin{itemize}

\setlength{\itemsep}{0pt}
    \item {\bf Stimulus terms}: As mentioned in \S\ref{sec:overview}, a subset of the relation bits are related to input bits or constants; the stimulus terms appear as $(r_i \leftarrow x_j)$ where $r_i\in R, x_j\in X\cup \{0,1\}$.
    \item {\bf Latent terms}: These clauses establish a relation among the relation bits; the variables $v$ in these terms are drawn from the relation bits $R$, i.e. $v \in R$.
\end{itemize}
For example, in \figref{key} the terms $(r_0 \leftarrow x_1)$, $(r_1 \leftarrow x_2)$, and $(r_2 \leftarrow rand)$ are stimulus terms, while $(r_3 \leftarrow (r_0 \land r_1))$ and $(r_4 \leftarrow (r_1 \oplus r_2))$ are latent terms.

\paragraph{Budget.} As the key relation may need to be implemented within a limited hardware budget, our synthesis imposes a hard threshold on its size. The threshold could directly capture the hardware constraints for implementing the key relation (e.g., the estimated number of gates or ports) or indirectly indicate the complexity of the key relation (e.g., number of relation bits, propositional terms, or latent terms).

\subsection{Lock and Key Inference}\label{sec:keyinference}

\begin{algorithm}[t]
\SetAlgoLined
 $\key \gets \emptyset$ \;
 $\lckd \gets \ckt$ \;
 $done \gets$~{\tt False} \;
 \While {{\bf not} $done$}{\label{line:start_loop}
    $E \gets SelectExpr(\lckd)$\;
    $H, \hat{E} \gets Synthesize(\key,\lckd,E)$ \;\label{line:synth}
    $\key' \gets \key \cup H$ \;
    \eIf {$Budget(\key') \leq T$}{ \label{line:budget}
        $\key \gets \key'$ \;
        $\lckd \gets \lckd[\{\hat{e}_i/e_i \ |$\\ \hspace{6mm}$\ e_i \in E_i, \hat{e}_i \in \hat{E} \}]$ \;\label{line:lockedckt}
    }{\label{line:final_tests}
        $q \gets CheckSec(\key, \lckd)$ \;\label{line:attack}
        \eIf {$q$}{\label{line:attackcondcheck}
            $done \gets$~{\tt True} \;
        }{
            $\key \gets \emptyset$ \;
            $\lckd \gets \ckt$ \;
        }
    }
 }\label{line:end_loop}
 \Return $\lckd, \key$\;
 \caption{HOLL($\ckt, T, Q$)} \label{alg:HOLL}
\end{algorithm}

\algref{HOLL} outlines our algorithm for inferring the key relation and the locked circuit. The algorithm accepts an unlocked design $Y\bimp \ckt(X)$ and a budget $T$.
\subsubsection{Main Algorithm.}
The algorithm iterates, increasing the complexity of the key relation, till the budget $T$ is reached (Lines~\ref{line:start_loop}-\ref{line:end_loop}). In every iteration, the algorithm selects a set of suitable expressions $E$ for locking, uses our synthesis procedure to extract a set of additional latent terms $H$ for the key relation, and produces the mutated expressions $\hat{e}_i$ for each expression $e_i \in E$ (Line \ref{line:synth}). If the additional synthesized relations keep the key relation within the budget $T$ (Line~\ref{line:budget}), the mutated expressions are replaced for $e_i \in E$ (Line~\ref{line:lockedckt}).
HOLL verifies that the solution meets the two objectives of correctness and security (\S\ref{sec:overview}). We handle correctness in the $Synthesize$ procedure of \algref{HOLL} and security in Lines \ref{line:attack}-\ref{line:attackcondcheck} of the same algorithm. The $CheckSec()$ procedure verifies if the synthesized (locked) circuit and key relations satisfy the security condition (Eqn \ref{eqn:security}). If $CheckSec()$ returns \texttt{True}, the key relation $\key$ and the locked circuit $\lckd$ are returned; otherwise, synthesis is reattempted.
\subsubsection{Correctness.}
\begin{figure}[t]
          \begin{subfigure}[b]{0.4\textwidth}
            $(r_0 \leftarrow x_1)$, \\
            $(r_1 \leftarrow x_2)$, $(r_2 \leftarrow x_0)$,\\
            $(r_5 \leftarrow (r_0 \land r_1))$,\\
            $(r_4 \leftarrow ($\fcolorbox{black}{red!20}{$(r_0 \land r_1)$}$\land r_2))$,\\
            $(r_3 \leftarrow ($\fcolorbox{black}{red!20}{$(r_0 \land r_1)$}$\lor r_2))$
          \caption{Without optimization} \label{fig:unoptkey}
          \end{subfigure}
          \hfill
          \begin{subfigure}[b]{0.4\textwidth}
           $(r_0 \leftarrow x_1)$, \\$(r_1 \leftarrow x_2)$,\\
           $(r_2 \leftarrow x_0)$, \\$(r_5 \leftarrow (r_0 \land r_1))$,\\
           $(r_3 \leftarrow ($\fcolorbox{black}{green!20}{$r_5$}$\lor r_2))$,\\
           $(r_4 \leftarrow ($\fcolorbox{black}{green!20}{$r_5$} $\land r_2))$
          \caption{With optimization} \label{fig:optkey}
          \end{subfigure}
          \captionof{figure}{Key relations generated without and with optimization. In the optimized version, $(r_0 \land r_1)$ in $r_3$ and $r_4$ is replaced by reusing the term $r_5$ similar to multi-level logic optimizations.}\label{fig:optrelation}
\end{figure}

    {\toolname} attempts to synthesize (via the $\textit{Synthesize}$ procedure) a key relation $\key$ and a set of locked expressions $\hat{e}_i$ such that the circuit is guaranteed to work correctly for all inputs given to $\key$; this requires us to satisfy:
        \[ \exists \key, \hat{e}_1, \dots, \hat{e}_n.\ \forall X.\ \forall R.\ (\key(X,R)\implies{\lckd}(X,R) = \ckt(X))\tag{5}\label{eqn:algcorrectness}\]
        where $\lckd \equiv \ckt[\hat{e}_1/e_1, \dots, \hat{e}_n/e_n]$, for $e_i \in (E\subseteq \ckt)$. In other words, we attempt to synthesize a set of modified expressions $\hat{E}$ that, once replaced the selected expressions in $E$, produces a semantically equivalent circuit as the original circuit if the relation $\key$ holds.
We solve this synthesis problem via counterexample-guided inductive synthesis (CEGIS)~\cite{alur2015}.
We provide a domain-specific language (DSL) in which $\key$ and $\hat{e}_i$ are synthesized. CEGIS generates candidate solutions for the synthesis problem and uses violations to the specification (i.e. the above constraint) to guide the search for suitable \textit{programs} $\key$ and $\hat{e}_i$.
\newline \indent
        A problem with the above formulation is illustrated in \figref{optrelation}: the key relation in \figref{unoptkey} uses \texttt{5} gates without reusing expressions, ``wasting" hardware resources. \figref{optkey} shows an optimized key relation that reuses the response bit $r_5$, allowing an implementation with only \texttt{3} gates. To encourage subexpression reuse, we solve this optimization problem:
        \begin{multline}
            \underset{\textsf{budget($\key$)}}{\textsf{argmin}} \exists \key, \hat{e}_1, \dots, \hat{e}_n.\ \forall X.\ (\forall R.\ \key(X,R)  \implies {\lckd}(X,R) = \ckt(X))\tag{6}\label{eqn:argmin}
        \end{multline}
        where $\lckd \equiv \ckt[\hat{e}_1/e_1, \dots, \hat{e}_n/e_n]$, for $e_i \in E\subseteq \ckt$.

\subsubsection{Security.}
    The security objective requires that the locking (i.e., the key relation $\key$ and the locked expressions) is non-trivial; that is. there exists some relation $\psi'$:$\psi' \neq \psi$ for which the circuit is not semantically equivalent to the original design:
        \[ \exists \key', \key'\neq \key, \textrm{ s.t. } \ \exists X.\ (\exists R.\ \key'(X,R)\land {\lckd}(X,R) \neq \ckt(X))\tag{7}\label{eqn:securitycheck}\]

    The above constraint is difficult to establish while synthesizing $\psi$; it requires a search for a different relation $\psi'$ that makes $\lckd$ semantically distinct from $\ckt$ while $\psi$ maintains semantic equivalence. Instead, we use a two-pronged strategy:

        \begin{itemize}
        \setlength{\itemsep}{0em}
            \item We carefully design the DSL used to synthesize $\key$ and $\hat{e}_i$ to reduce the possibility they generate trivial relations;
            \item After obtaining $\key$ and $\lckd$, we attempt to synthesize an alternative relation $\key'$ (using \ref{eqn:wrongKey}) such that $\lckd$ is not semantically equivalent to $\ckt$, ensuring that $\key$ and $\lckd$ do not constitute a trivial locked circuit.
        \end{itemize}
\begin{align}
\exists \key'. \ \exists X,R'.\ \ckt(X) \neq \lckd(X, R') \land \key'(X, R')\tag{8}\label{eqn:wrongKey}
\end{align}

The procedure $CheckSec(\key, \lckd)$ (Algorithm \ref{alg:HOLL}, Line \ref{line:attack}) implements the above check (Eqn. \ref{eqn:wrongKey}).

\begin{theorem}
If \algref{HOLL} terminates, it returns a correct (Eqn. \ref{eqn:correctness}) and secure (Eqn. \ref{eqn:security}) locked design.
\end{theorem}
\begin{proof}
The proof follows trivially from the design of the \textit{Synthesize} (in particular, Eqn. \ref{eqn:algcorrectness}) and \textit{CheckSec} (in particular, Eqn. \ref{eqn:wrongKey}) procedures (respectively).
\end{proof}

\subsection{Expression Selection}\label{sec:exprsel}
\begin{figure}[t]
\vspace{-70pt}
\centering
\includegraphics[height=9.2cm]{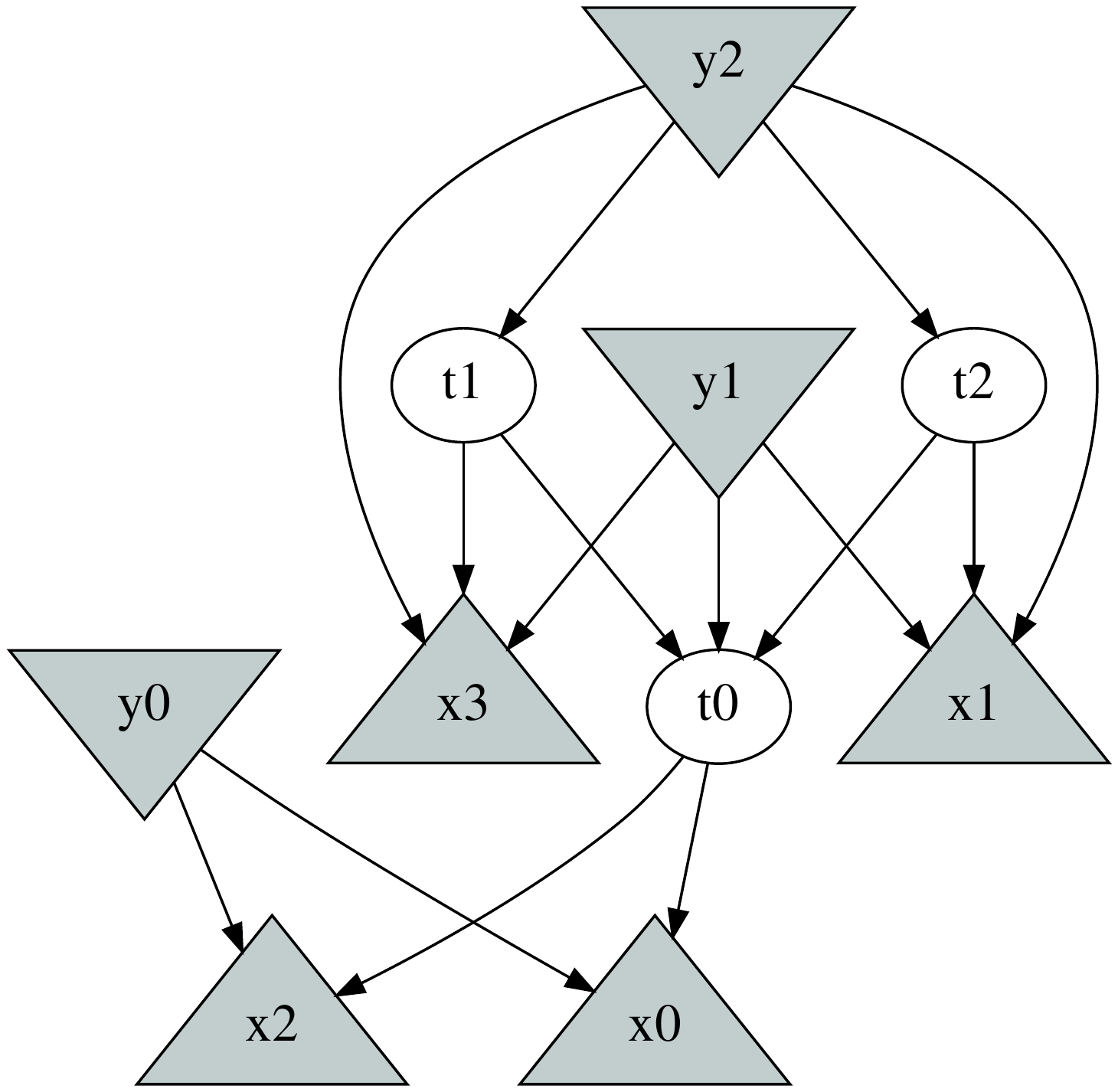}
\caption{Dependency graph for the expressions in \figref{adder}.}
\label{fig:dependencygraph}
\end{figure}
{\toolname} constructs the dependency graph~\cite{ferrante1987} $(V, D)$ for expression selection, where nodes $V$ are circuit variables. A node $v\in V$ represents an expression $e$ such that $v$ is assigned the result of expression $e$, i.e. $(v\gets e)$. The edges $D$ are dependencies: the edge $v_1 \rightarrow v_2$ connects the two nodes $v_1$ to $v_2$ if variable $v_1$ depends on variable $v_2$. The tree is rooted at the output variables and the input variables appear as leaves.

For example, \figref{dependencygraph} shows the dependency graph for the circuit in \figref{adder}. Triangles represent input ports ($x_0,x_1,x_2,x_3$) while inverted triangles represent output ports ($y_0,y_1,y_2$).

Our variable selection algorithm has the following goals:
\begin{enumerate}
\setlength{\itemsep}{0pt}
    \item {\bf Ensure expression complexity}: The algorithm selects an expression $e_z$ as a candidate for locking only if the depth of the corresponding variable $z$ in the dependency graph lies in a user-defined range [L, U] to create a candidate set $Z$. The lower threshold $L$ assures the expression captures a reasonably complex relation over the inputs, while the upper threshold $U$ ensures the relation is not too complex to exceed the hardware budget. The algorithm starts by randomly selecting a variable $z_0\in Z$ from this set.

    \item {\bf Encourage sub-expression reuse in key relation}: We attempt to select multiple ``close" expressions; for the purpose, the algorithm randomly selects variables $w_i\subseteq Z$ on which $z_0$ (transitively) depend on.

    \item {\bf Encourage coverage}: We select expressions for locking in a manner so as to \textit{cover} the circuit. To this end, interpreting $(V, D)$ as an undirected graph, we randomly select expressions $u_i \in Z$ that are \textit{farthest} from $z_0$, i.e. the shortest distance between $u_i$ and $z_0$ is maximized.
\end{enumerate}
Our algorithm first executes step (1), and then, indeterminately alternates between (2) and (3), till the required number of variables are selected. Let us use the dependency graph in \figref{dependencygraph} to show how the above algorithm operates:

\begin{itemize}
\setlength{\itemsep}{-0.5pt}
    \item Given the user-defined range \texttt{[1,3]}, we compose the initial candidate set $Z =\{y_0,t_0,t_1,t_2,y_1,y_2\}$.
    \item Let us assume we randomly pick the expression for $y_2$. Now, $y_2$ depends on expressions $t_0$, $t_1$ and $t_2$ ($\{t_0,t_1,y_2\}\subseteq Z$) [{\em Rule~1}].
    \item We randomly choose new expressions to lock/transform from $\{t_0,t_1,y_2\}$. For example, we select $t_2$ and $t_0$---[{\em Rule~2}].
    \item We find $y_0$, which is the furthest expression from $t_0,t_2,y_2$ in $Z$. We select to lock the set of expressions \{$y_1,y_2,t_0,t_2$\}---[{\em Rule~3}].
\end{itemize}

%% file: content/attack.tex
\section{HOLL: Implementation and Optimization}\label{sec:toolimplimentation}
\subsubsection{Implementation.}
We implemented {\toolname} in Python, using {\sc Sketch}~\cite{sketch} synthesis engine to solve the program synthesis queries. We used {\sc Berkeley-ABC}~\cite{abc} to convert the benchmarks into Verilog and {\sc PyVerilog}~\cite{Takamaeda:2015:ARC:Pyverilog}, a Python-based library, to parse the Verilog code and generate input for {\sc Sketch}.
We use the support for optimizing over a synthesis queries provided by {\sc Sketch} to solve Eqn.~\ref{eqn:argmin}. Algorithm~\ref{alg:HOLL} may not terminate; our implementation uses a timeout to ensure termination.

\subsubsection{Domain Specific Language.}\label{sec:dsl}
We specify our domain-specific language for synthesizing our key relations and locked expressions. The grammar is specified as generators in the {\sc Sketch}~\cite{sketch} language.
The grammar for the key relations and locked expressions is as follows:
\begin{align*}
\B{G} &::= r \leftarrow x\ \vert\ r \leftarrow r \B{Bop} r\  \vert \  r \leftarrow \B{Uop} r\  \vert\  r \leftarrow r\\
\B{Bop} &::= or \  \vert\  and \  \vert \ xor\\
\B{Uop} &::= not
\end{align*}
The rule $\B{G} ::= r \leftarrow x$ is only present in the key relation grammar since the locked expressions have no input bits.

\subsubsection{Backslicing.}\label{sec:backslicingexprsel} To improve scalability, we use backslicing~\cite{venkatesh91} to extract the portion of the design related to the expressions selected for locking. For a variable $v_i$, the set of all transitive dependencies that can affect the value of $v_i$ is referred to as its \textit{backslice}. For example, in \figref{dependencygraph}, $backslice(t2)=\{t_0,x_0,x_1,x_2\}$.

Given the set of expressions $E$, we compute the union of the backslices of the variables in $E$, i.e. all expressions $B$ in the subgraph induced by $e \in E$ in the dependency graph; we use $B \subseteq E$ for lock synthesis.

Backslicing tilts the asymmetrical advantage further towards the {\toolname} defense. The attacker cannot apply backslicing on the locked design since the dependencies are obscured, preventing the extraction of the dependency graph.

\subsubsection{Incremental Synthesis.}\label{sec:incrementalsynt}
Given a set of expressions $E$, the procedure $Synthesis$ in \algref{HOLL} creates a list of relations $H$ and a new set of locked expressions $\hat{E}$. If the list of expressions is large, we select the expressions in the increasing order of their depth in the dependency graph. The lower the depth of the expression is, the closer it is to the inputs, and the simpler is the expression. Selecting an expression with the lowest depth first (say $e_1$) ensures that other expressions ($e_j$) depending on this expression can use the relations $H$ generated during synthesis of $\hat{e}_1$. This also makes synthesizing the other expressions easier as the current relation has some sub-expressions on which the new relations can be built.
\section{SynthAttack: Attacking HOLL with Program Synthesis}\label{sec:attack}
As HOLL requires inference of relations and not values, existing attacks designed for logic locking do not apply.
We design a new attack, SynthAttack, that is inspired by the SAT attack~\cite{7140252} for logic locking and counterexample-guided inductive program synthesis (CEGIS)~\cite{solar-lezama2013}.

\subsection{The SynthAttack Algorithm}
\begin{algorithm}[t]
\SetAlgoLined
    $i \gets 0$ \;
    $Q_0\gets \top$ \;\label{line:Q0}
    \While{$\texttt{True}$}{\label{line:QloopCond}
      $X' \gets Solve_X(Qi \land $\label{line:distingInput}\label{line:distingInput} \\ \hspace{5mm}$(\lckd(X,R_1) \neq \lckd(X,R_2)) \land$ \\ \hspace{5mm}$ \key_1(X,R_1) \land \key_2(X,R_2))$ \;
      \If{$X' = \bot$}{\label{line:breakcond}
        \textbf{break}\;\label{line:break}
      }
      $Y' \gets \ocl(X')$ \;\label{line:distingOutput}
      $Q_{i+1} \gets Q_i \land (\lckd(X',R_1^{i}) \bimp $ \label{line:updatedQ} \\ \hspace{5mm}$ Y') \land (\lckd(X',R_2^{i}) \bimp Y')$ \\ \hspace{5mm}$\land\ \key_1(X',R_1^{i}) \land $ \\ \hspace{5mm}$ \key_2(X',R_2^{i})$ \;
      $i \gets i+1$ \;
    }
    $\key_1, \key_2 \gets Solve_{\key_1, \key_2}(Q_i)$ \;\label{line:correctKey}
    \Return $\key_1$\;
    \caption{SynthAttack($\lckd$, $\ocl$)}\label{alg:SynthAttack}
\end{algorithm}
SynthAttack runs a counterexample-guided inductive synthesis loop: it accumulates a set of examples, $\Lambda$, by discovering some ``interesting" inputs, and uses an activated circuit as an input-output oracle to compute its corresponding response. These examples, $\Lambda$, are used to constrain the space of the candidate key-relations. SynthAttack, then, uses a verification check to confirm if the collected examples are sufficient to synthesize a valid key-relation, i.e. one that provably activates the correct functionality on the locked circuit. Otherwise, the counterexample from the failed verification check is identified as an ``interesting input" to be added to $\Lambda$, and the algorithm repeats. The verification check is based on the existence of a \textit{distinguishing input} \S\ref{sec:overview}.

If there does not exist any distinguishing input for the locked circuit $\lckd$, then $\lckd$ has a unique semantic behavior---and that must be the correct functionality. Any key-relation that satisfies the counterexamples (distinguishing inputs) generated so far will be a valid candidate for the key relation. An inductive synthesis strategy based on distinguishing inputs allows us to quickly converge on a valid realization of the key-relation as each distinguishing input disqualifies many potential candidates for the key relation. Note that (as we illustrate the following example) there may be multiple, possibly semantically dissimilar, realizations of a key-relation that enables the same (correct) functionality on the locked circuit.

SynthAttack is outlined in \algref{SynthAttack}: the algorithm accepts the design of the locked circuit ($\lckd$) and an activated circuit $\lckd(\key)$ (the locked circuit $\lckd$ activated with a valid key-relation $\key$). The notation $\ocl$ indicates that this activated circuit can only be used as an input-output oracle, but cannot be inspected.

SynthAttack runs a counterexample-guided synthesis loop (Line \ref{line:QloopCond}). It checks for the existence of a distinguishing input in Line~\ref{line:distingInput}: if no such distinguishing input exists, it implies that the current set of examples is sufficient to synthesize a valid key-relation. So, in this case, the algorithm breaks out of the loop (Line~\ref{line:breakcond}-\ref{line:break}) and synthesizes a key-relation (Line~\ref{line:correctKey}), that is returned as the synthesized, provably-correct instance for the key relation.

If there exists a distinguishing input $X'$ (in Line~\ref{line:distingInput}), the algorithm uses the activated circuit to compute the expected output $Y'$ corresponding to $X$ (Line~\ref{line:distingOutput}). This new counterexample $(X', Y')$ is used to \textit{block} all candidate key-relations that lead to an incorrect behavior, thereby reducing the potential choices for $\key_1$ and $\key_2$. The loop continues, again checking for the existence of distinguishing inputs on the updated constraint for $Q_i$. The theoretical analysis of SynthAttack is provided in \S\ref{sec:analysisSynthAttack}.

The algorithm only terminates when it is able to synthesize a provably valid key-relation, that allows us to state the following result.

\begin{theorem}\label{thm:SynthAttack}
    Algorithm \ref{alg:SynthAttack} will always terminate, returning a key-relation  $\key_1$ such that $\lckd(\key_1)$ is semantically equivalent to $\lckd(\key)$, where $\key$ is the ``correct" relation hidden by HOLL.
\end{theorem}
\begin{proof}
The number of loop iterations in the above procedure is upper-bounded by the number of distinguishing inputs and the number of distinguishing inputs is upper-bounded by the number of possible inputs (that is upper-bounded by $\texttt{2}^n$ where $n$ is the number of input of $\lckd$). Hence, termination of the procedure is guaranteed. That the algorithm only returns a valid key-relation is ensured by the check for the existence of distinghuishing inputs at Line~\ref{line:breakcond} in Algrorithm~\ref{alg:SynthAttack}.
\end{proof}

\begin{figure}[t]
\begin{minipage}{0.45\textwidth}
\{$(r_0 \bimp x_1)$,\\
$(r_1 \bimp x_2)$,\\
$(r_2 \bimp r_0)$,\\
$(r_3 \bimp r_2 \land r_1)$,\\
$(r_4 \bimp \lnot r_2 \land r_1)$\}\\
\captionof{figure}{Key relation generated by SynthAttack.} \label{fig:synthkey}
\end{minipage}
\hfill
\begin{minipage}{0.45\textwidth}
\centering
\captionof{table}{Distinguishing inputs.}\label{tab:adderSynthAttack}
\begin{tabular}{ c | c }
\hline
$X$ &$Y$ \\\hline\hline
1101 &100  \\\hline
0001 &001  \\\hline
0101 &010  \\\hline
0111 &100  \\\hline
1001 &011  \\\hline
0011 &011  \\\hline
\end{tabular}
\end{minipage}
\end{figure}
\paragraph{Example. } Executing SynthAttack on \figref{lockadder} iteratively generates six distinguishing inputs as shown in \tabref{adderSynthAttack}.
\figref{synthkey} shows the key relation synthesized by SynthAttack. This key-relation (\figref{synthkey}) is {\bf not} semantically equivalent to the hidden key-relation that was computed and hidden by HOLL (\figref{key}). This shows that there may exist multiple valid candidates for the key-relation that all evoke the same functionality on the locked design. For example, $X=\texttt{0100}$ generates $r_4=\texttt{1}$ for the key relation in \figref{key} but $r_4=\texttt{0}$ for \figref{synthkey}; however, the output of the locked circuit remains the  same in both cases ($Y=\texttt{001}$).

\section{Analysis of SynthAttack}\label{sec:analysisSynthAttack}
The \textit{attack resilience} of HOLL against SynthAttack can be captured by the total execution time ($T$) taken to generate a valid key-relation:
\begin{align}
T \geq \sum\limits_{i=1}^{n} t_i \tag{9}\label{eqn:attackres}
\end{align}
where $n$ is the total number of distinguishing inputs (counterexamples) generated by the algorithm and $t_i$ is the verification time (Algorithm~\ref{alg:SynthAttack}, Line~\ref{line:distingInput}) in the $i^{th}$ iteration of the loop (i.e. the time taken to generate the $i^{th}$ distinguishing input or prove unsatisfiability). We say that a locked circuit, $\lckd$, is attack resilient against SynthAttack if $T$ is \textit{sufficiently} large. One can increase $T$ by either increasing $n$ or $t_i$ (or both).

    \paragraph{Verification Time ($t_i$). } Line~\ref{line:updatedQ} in the Algorithm shows that the formula $Q_i$ grows in each iteration (i.e. with $i$) by more than two times the size of the locked circuit $\lckd$. As the worst-case time for program synthesis increases exponentially with the size of the formula, successive iterations tend to get quite hard. This is also confirmed by our experimental results (\S\ref{sec:attackreselience}). 
    \paragraph{Number of iterations (n). } The number of iterations (n) depends on the number of distinguishing inputs accumulated till the verification check returns \textsc{Unsat}. The number of distinguishing inputs is upper-bounded by the minimum of the number of possible inputs (that is exponential in the number of input ports) and the number of possible functional relations that can be candidates for the key-relation (which is exponential in the number of relation bits). We attempt to increase $n$ by controlling the hardware budget for the key-relation and the expression selection heuristics (\S\ref{sec:exprsel}). Unfortunately, it is difficult to establish a non-trivial lower-bound; it depends on the design of the locked circuit and that of the key relation.

\section{Hardware Implementation}\label{sec:hardware_architecture}
We discuss how to hide the {\em key relation} from the untrusted foundry in a provably secure way. We need a solution to store the key relation as a {\em program} (not simply key bits), which must be hidden during chip fabrication. For this, we borrow concepts from programmable logic circuits that can implement any relation provided by a customer.
The use of programmable devices is common in the hardware industry~\cite{efpgaMarket}.

    Depending on the key relation, we can use different solutions,
    based on the size and complexity of the key relation:

    We can implement a simple relation with an embedded {\em Programmable Array Logic} (PAL). A PAL device includes a set of elements that can be programmed by the users with specialized machines. The transistors are arranged in a ``fixed-OR, programmable-AND`` plane to implement the well-known {\em sum-of-products} (SOP). In Boolean algebra, any relation $\mathcal{R}$ can be expressed with an SOP form, also known as {\em disjunctive normal form}. The SOP form is then a canonical representation of the relation $\mathcal{R}$. Most common logic synthesis tools can easily extract the SOP form associated with each relation bit $r_i$ to program the PAL. For example, given the relation $r_3 \leftarrow r_0 \land r_1 \lor r_2$ of \figref{unoptkey}, the corresponding SOP minterms are shown in \figref{pla_sop}, while \figref{pla_ckt} shows its implementation with a single-output PAL.

    \begin{figure}[t]
    \begin{subfigure}[b]{0.4\textwidth}
    $a =~\overline{r_0}\cdot\overline{r_1}\cdot r_2$ \\
    $b =~\overline{r_0}\cdot r_1\cdot r_2$  \\
    $c =~r_0\cdot\overline{r_1}\cdot r_2$ \\
    $d =~r_0\cdot r_1\cdot\overline{r_2}$ \\
    $e =~r_0\cdot r_1\cdot r_2$  \\
    $r_3 = a + b + c + d + e$\\
    \caption{Minterms and SOP form of the relation $r_3$.}\label{fig:pla_sop}
    \end{subfigure}
    \begin{subfigure}[b]{0.4\textwidth}
    \includegraphics[width=\textwidth]{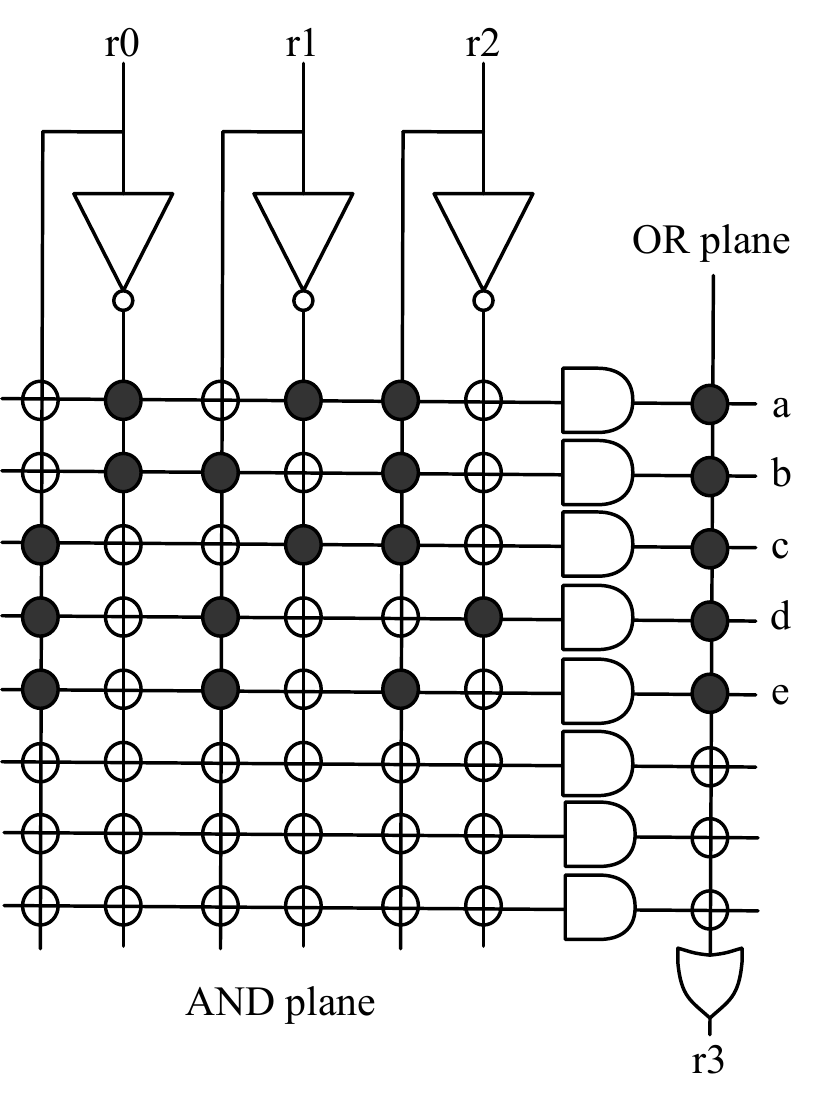}
    \caption{PAL implementation.}\label{fig:pla_ckt}
    \end{subfigure}
    \caption{Example of implementation for a key relation generated by HOLL.}\label{fig:pla}
    \end{figure}

    PAL devices are fast and simple to program, but can only support small circuits. For example, the simple expression above has three literals but requires five minterms (so taking five gates in the AND plane). If the designer requires more complex relations that require more logic resources and even also storing capabilities, she can use an embedded {\em Field-Programmable Gate Array} (eFPGA). An eFPGA is a device composed of several configurable logic blocks, flip-flop registers, and heterogeneous resources interconnected with a reprogrammable fabric. In this case, the designer must use specific toolchains from the eFPGA provider to generate the proper configuration bitstream, i.e. the sequence of bits needed to specify the desired functionality~\cite{8892171}.

    In both cases, implementing a relation with programmable devices requires translating the specification into a configuration for all cells and/or switches. The more configurable cells, the more complex key relations we can implement. Without the correct configuration of all elements, the device cannot reproduce the correct relation. The attacker cannot know how many and which gates are used in the programmable device to implement the key relation. So, these devices can be a valid solution to securely store the key relation.

    We end this section with an important remark: the \textit{eFPGA configuration} can be represented as a sequence of \textit{configuration} bits. However, such a configuration bit-sequence represents a \textit{function}, as compared to the key bits in traditional logic locking that simply represent a \textit{value}. While any function can be encoded as a bit-sequence~\cite{clift2020, petersen2019}, inferring a function bit-sequence is a computationally much harder problem (a second-order problem) as compared to inferring a value bit-sequence (a first-order problem). In the case of functions, only a complete configuration bit recovery allows the attacker to unlock the function. The complexity of inference lies in the complete interpretation of the entire bit-sequence, not in its simple or partial representation.

    The circuit configurations from HOLL are one to two orders larger than the typical sizes of key bits in logic locking even when the key relation has a limited size~\cite{Deepak19}.
    Further, as the attacker does not know how the gates are used in the programmable device to implement the key relation, an attacker will have to consider all possible gate allocations, resulting in an impractical search over a huge design space.

%% file: content/experimentation.tex
\section{Experimental Evaluation} \label{sec:experimentaleval}

We selected \texttt{100} combinational benchmarks from ISCAS'85~\cite{iscas85} and MCNC~\cite{yang1991} and report the time for program synthesis and the overhead after applying our locking method. For long running experiments, we select a subset of \texttt{10} randomly selected benchmarks 
where, number of input ports range between 16 and 256, output ports range range between 7 and 245, \texttt{AND} gates range between 135 and 4174.
\tabref{benchmarks} summarizes the ten selected benchmarks. We use {\sc Berkeley-ABC} to get the number of \texttt{AND} gates after we flatten the circuit using \texttt{strash} command.
\begin{table}[t]
\centering
\caption{Our selected benchmark set.}\label{tab:benchmarks}
\begin{tabular}{ c | c c c | c}
\hline
    & \multicolumn{3}{|c|}{\bf Ports} &    \\
   {\bf Bench} & {\bf \#IN} & {\bf \#OUT} & {\bf \#Total} & {\bf \#Gates} \\ \hline
     {\tt al2} &16 &47 &63 &135 \\
     {\tt cht} &47 &36 &83 &185 \\
     {\tt C432} &36 &7 &42 &212  \\
     {\tt C880} &60 &26 &86 &330 \\
     {\tt i9} &88 &63 &151 &682 \\
     {\tt i7} &199 &67 &266 &763 \\
     {\tt x3} &135 &99 &234 &816 \\
     {\tt frg2} &143 &139 &282 &1,233 \\
     {\tt i8} &133 &81 &214 &2,175 \\
     {\tt des} &256 &245 &501 &4,174 \\
    \hline
\end{tabular}
\end{table}

For our experiments, we use the number of \textit{relation} terms in the key relation as a proxy for its complexity. We define our \textit{budget} in terms of it and use a budget threshold for the key relation in the range \texttt{[12-14]} latent terms and depth of expression selection in range \texttt{[2-4]}. We use a timeout of \texttt{20} minutes for locking and \texttt{4} days for attacks. We conduct our experiments on a machine with \texttt{32}-Core Intel(R) Xeon(R) Silver \texttt{4108} CPU @ \texttt{1.80}GHz with 32GB RAM.

For both HOLL and SynthAttack, we use the {\sc Sketch} synthesis tool. Since synthesis solvers are difficult to compare across different problem instances, we were wary of the case where the defender gets an edge over the attacker due to use of different tools. We create the attack-team-defence-team asymmetry by controlling the computation time: while the defender gets \texttt{20} minutes (\texttt{1200}s) to generate locked circuit, the attacker runs the attack for up to \texttt{4} days.

Our experiments aim to answer five research questions:
\begin{description}
\setlength{\itemsep}{0pt}
 \item[\bf RQ1.] What is the attack resilience of HOLL? (\S\ref{sec:attackreselience})
 \item[\bf RQ2.] How do impact expression selection heuristics affect attack resilience?~(\S\ref{sec:attackvsexprsel})
 \item[\bf RQ3.] What is the hardware cost for HOLL? (\S\ref{sec:hardwarecosteval})
 \item[\bf RQ4.] What is the time taken to synthesize the locked design and key-relation for HOLL? (\S\ref{sec:inferencetime})
 \item[\bf RQ5.] What are the impact of the optimizations for scalability (backslicing and incremental synthesis)? (\S\ref{sec:impactoptimization})
\end{description}
Here is a summary of our findings:

\noindent\fbox{
\parbox{0.95\textwidth}{
    {\bf Security.} The key relations can be recovered completely by the attacker via SynthAttack but only for small circuits with a small hardware budget. For medium and large designs, key relations are fast to obtain ($<$\texttt{1200}s) but cannot be recovered by our attack even within \texttt{4} days. This shows our defense is efficient while our attack is strong but not scalable.

   \vspace{4pt}{\bf Hardware Cost.} Our key relations with a budget of \texttt{12-14} latent terms have a minimal impact on the designs and the overhead reduces as the size of the circuit grows. On the largest benchmark, the area overhead is \texttt{1.2}\%. The corresponding configurations for programmable devices are small and provide high security.

   \vspace{4pt}{\bf HOLL Performance.} The HOLL execution time ranges between \texttt{8}s and \texttt{1001}s, with an average of \texttt{33}s for small, \texttt{17}s for medium, and \texttt{60}s for large designs for the budget of \texttt{8-10} latent terms. Our optimizations are crucial for the scalability of our HOLL defense (locking) algorithm: we fail to lock enough expressions in large circuits without these optimizations.
}
 }



\subsection{Attack Resilience}\label{sec:attackreselience}
We define \textit{attack resilience} of locked circuit, $\lckd$, in terms of time taken to obtain a key relation, $\key'$, such that $\lckd$ $\land$ $\key'$ is equivalent to original circuits, $\ckt$. 

\subsubsection{Attack time.}
\figref{synthattacktrend} shows the cumulative time spent till the $i^{th}$ iteration (y-axis) of the loop versus the loop counter~$i$, that is also the number of distinguishing inputs (samples) generated so far (x-axis). We show exponential \textit{trend} curves (as a solid \textcolor{magenta}{pink} line) to capture the trend in the plotted points while the data-points are plotted as \textcolor{blue}{blue} dots.
The plots show that the plotted points follow the exponential trend lines, illustrating that SynthAttack does not scale well, thereby asserting the resilience of HOLL.

SynthAttack failed to construct a valid key-relation for any of these ten designs within a timeout of \texttt{4} days. However, for small designs with lesser number of latent terms, SynthAttack was indeed able to construct a valid key-relation. For example, SynthAttack is able to construct a valid key-relation for the benchmarks {\tt al2} and {\tt i9} within \texttt{10} hours for {\tt \texttt{7}} latent terms, respectively (see Figure~\ref{fig:latentvstimeall}).

\begin{figure}[t!]
\centering
\begin{subfigure}[b]{0.45\textwidth}
\includegraphics[width=1\textwidth]{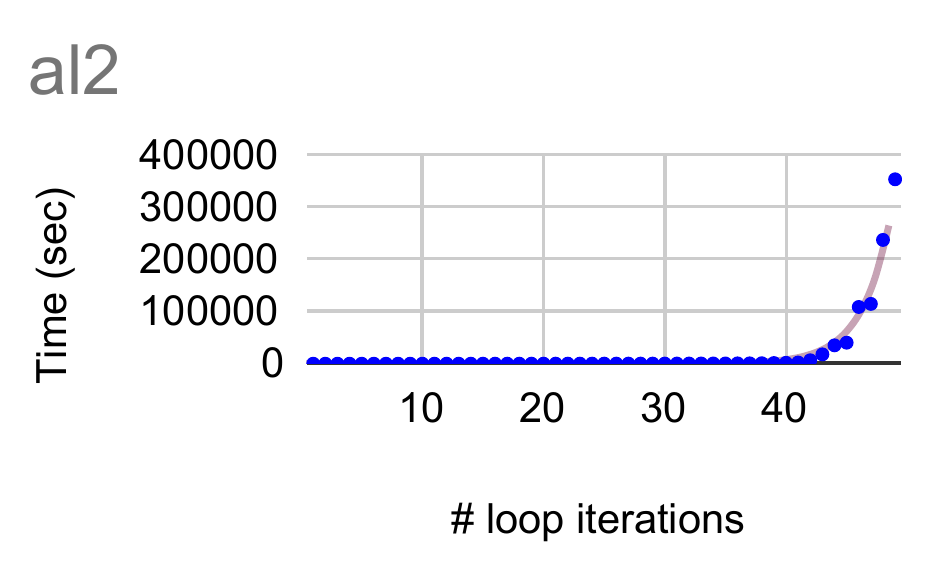}
\end{subfigure}
 \begin{subfigure}[b]{0.45\textwidth}
 \includegraphics[width=1\textwidth]{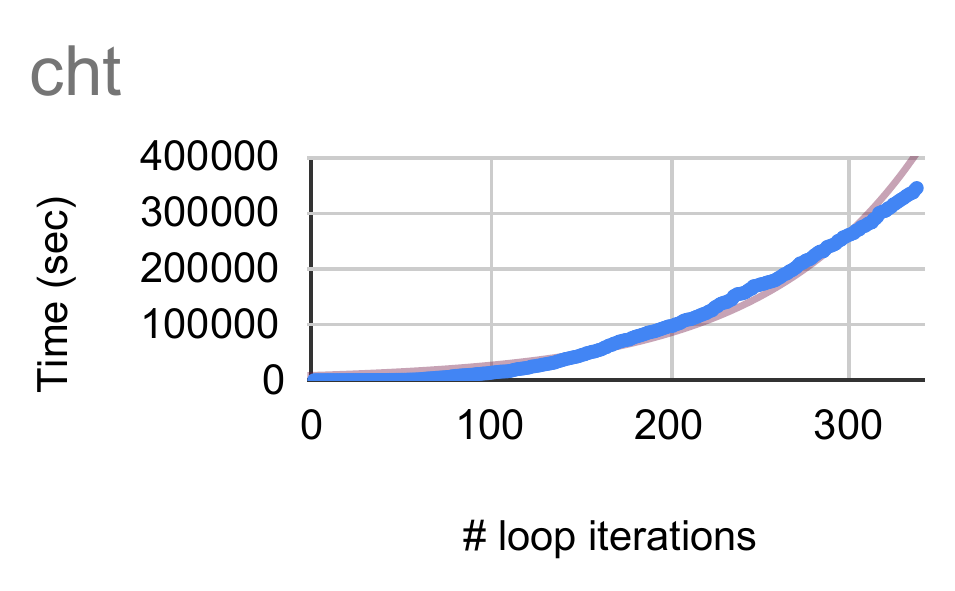}
 \end{subfigure}
\begin{subfigure}[b]{0.45\textwidth}
\includegraphics[width=1\textwidth]{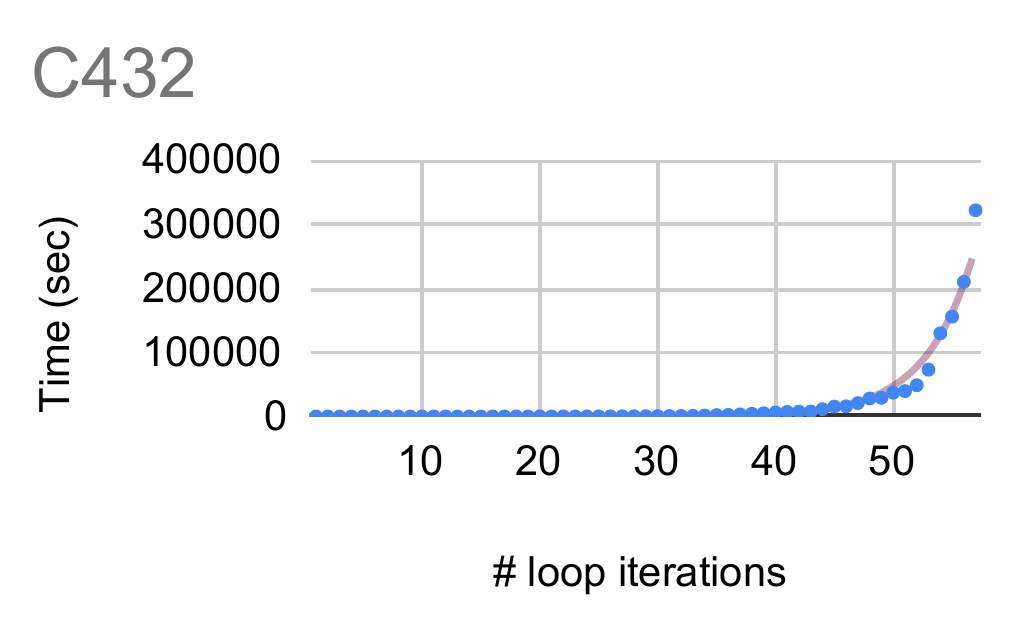}
\end{subfigure}
 \begin{subfigure}[b]{0.45\textwidth}
 \includegraphics[width=1\textwidth]{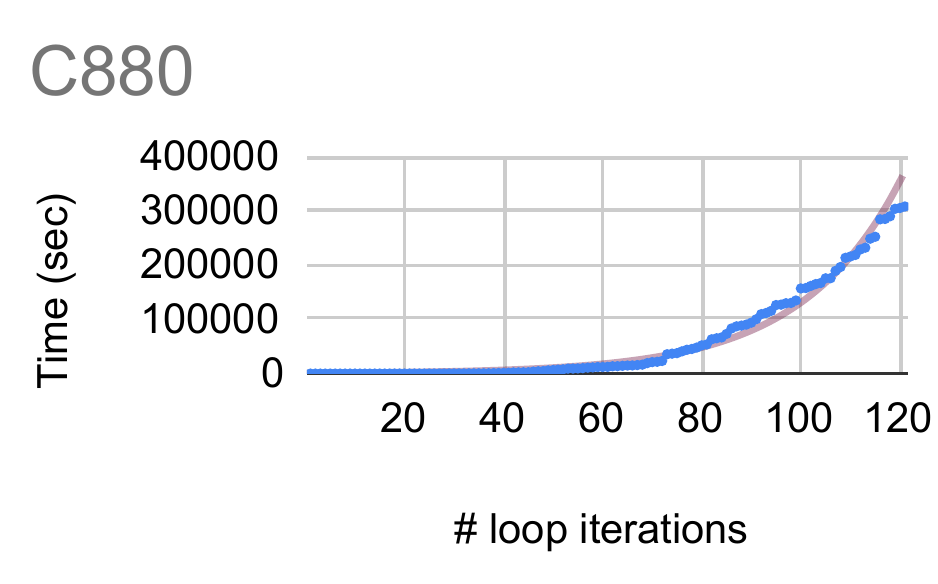}
 \end{subfigure}
\begin{subfigure}[b]{0.45\textwidth}
\includegraphics[width=1\textwidth]{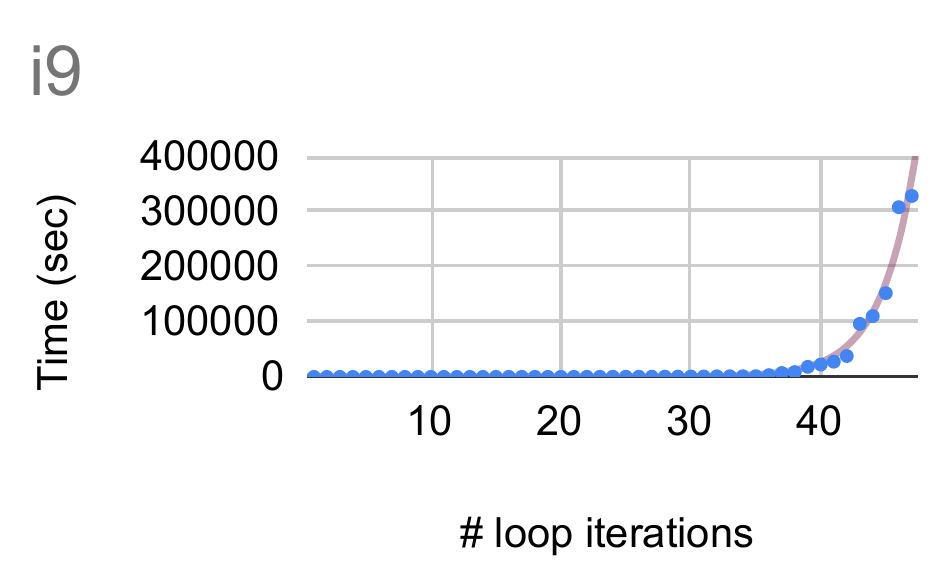}
\end{subfigure}
 \begin{subfigure}[b]{0.45\textwidth}
 \includegraphics[width=1\textwidth]{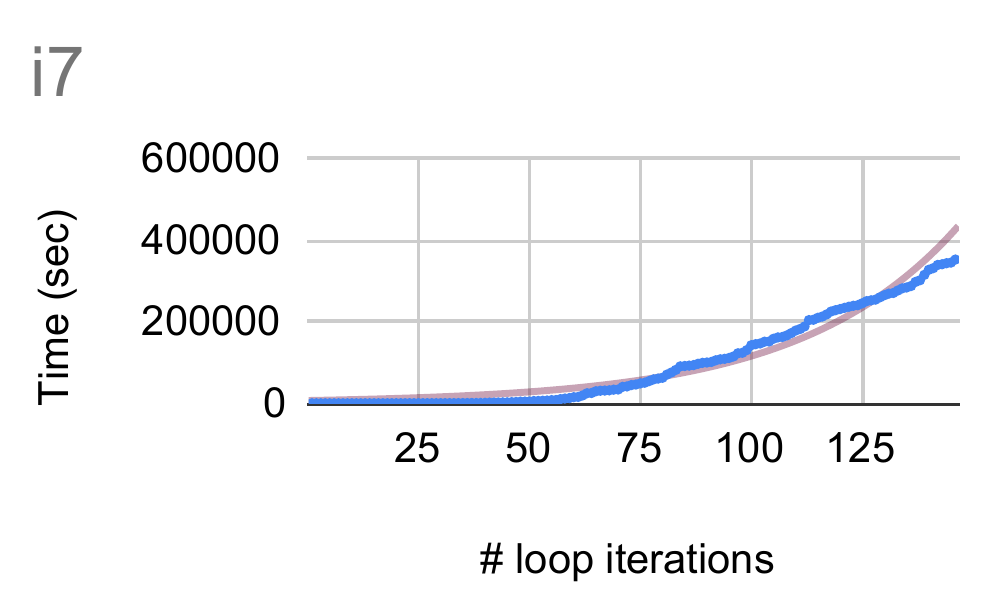}
 \end{subfigure}
 \begin{subfigure}[b]{0.45\textwidth}
 \includegraphics[width=1\textwidth]{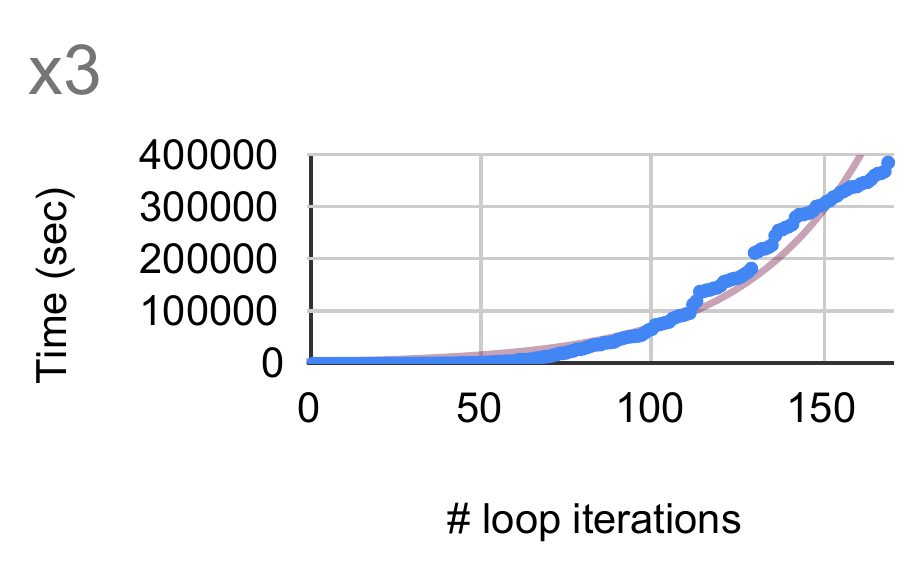}
 \end{subfigure}
 \begin{subfigure}[b]{0.45\textwidth}
 \includegraphics[width=1\textwidth]{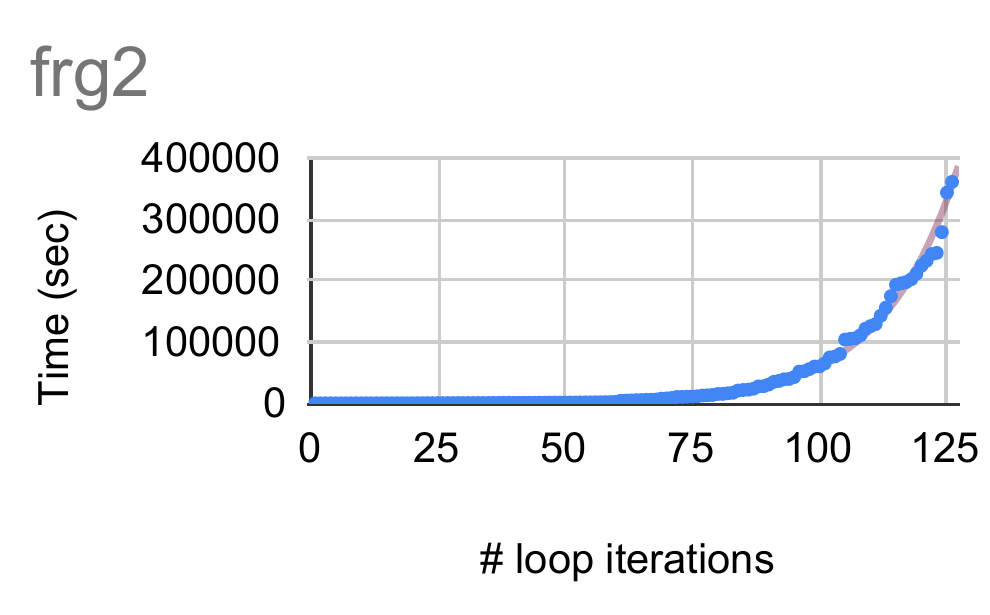}
 \end{subfigure}
 \begin{subfigure}[b]{0.45\textwidth}
 \includegraphics[width=1\textwidth]{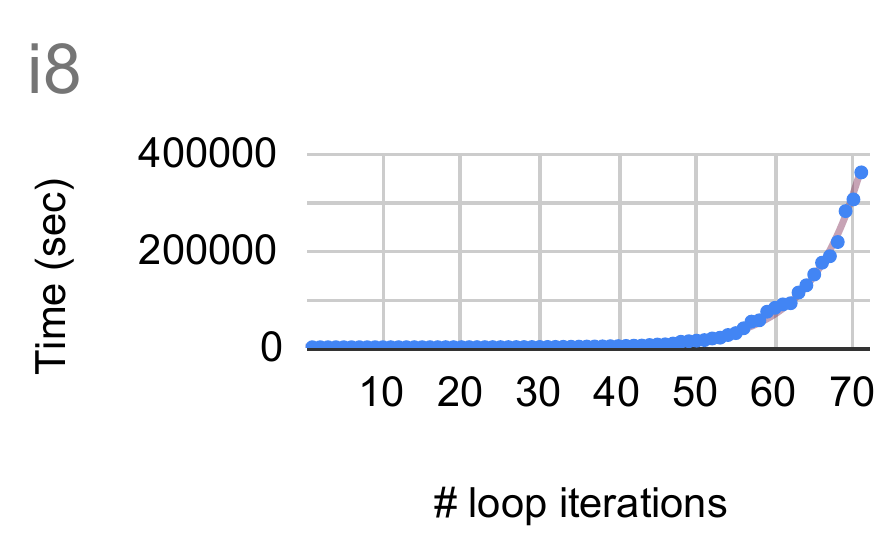}
 \end{subfigure}
\begin{subfigure}[b]{0.45\textwidth}
\includegraphics[width=1\textwidth]{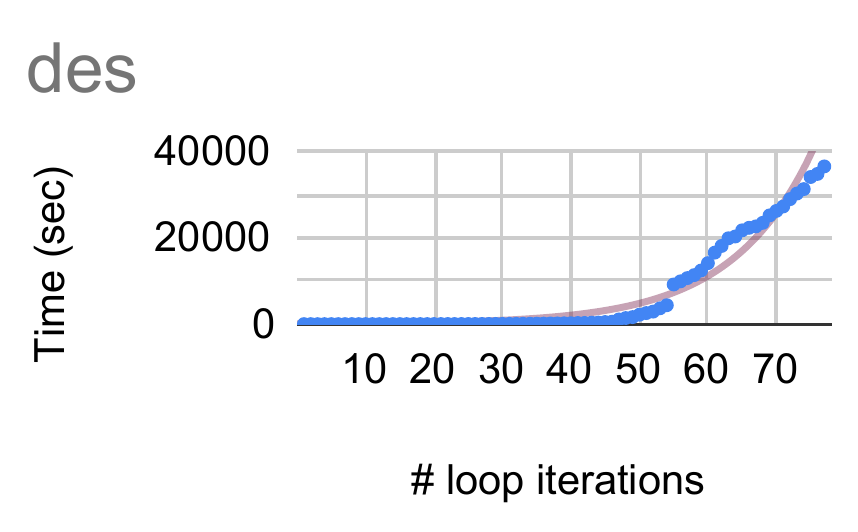}
\end{subfigure}
    \caption{figure}{Cumulative time for successive iterations of SynthAttack (best viewed in color)}\label{fig:synthattacktrend}
\end{figure}
\hfill
\begin{figure}[t]
\centering
\begin{subfigure}{0.45\textwidth}
    \centering
    \includegraphics[width=1\textwidth]{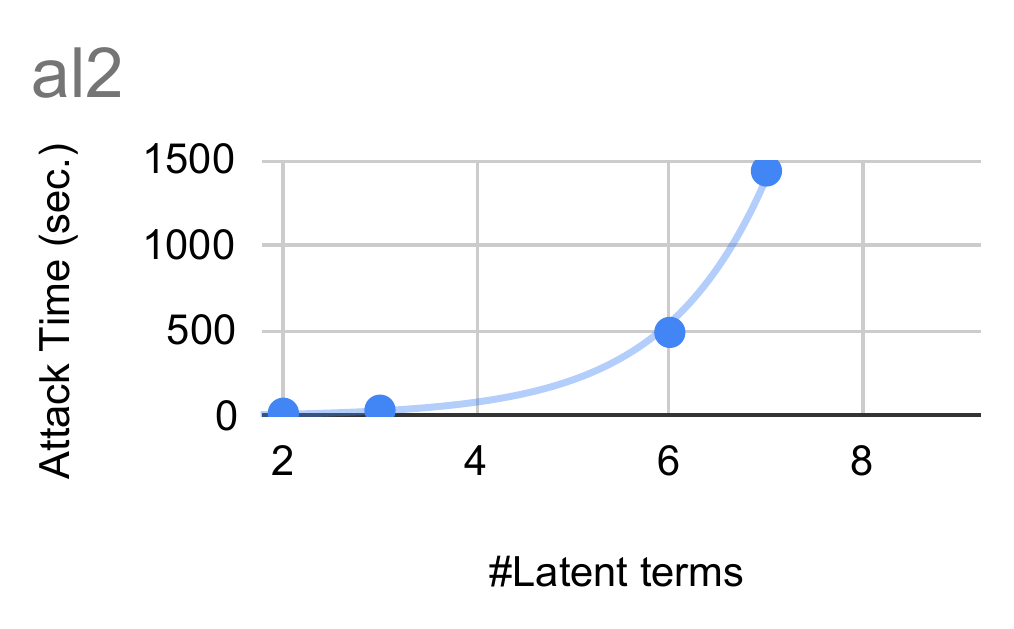}
\end{subfigure}
\begin{subfigure}{0.45\textwidth}
    \centering
    \includegraphics[width=1\textwidth]{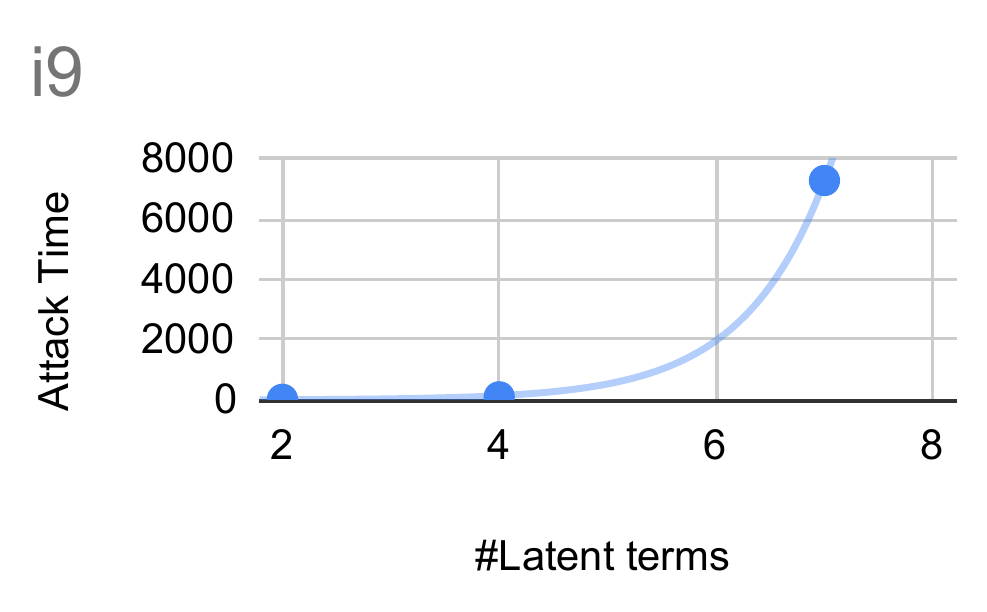}
\end{subfigure}
\caption{figure}{Attack time vs \#latent terms for {\tt i9} and {\tt al2}}\label{fig:latentvstimeall}
\end{figure}

\subsubsection{Attack resilience vs. number of latent terms.}

The complexity of the key relation increases with the number of relation bits. As shown in \figref{latentvstimeall} (for benchmarks {\tt al2} and {\tt i9}), the time required to break the locked circuit increases exponentially as the number of relation bits increases. We gave a timeout of \texttt{10} hours for this experiment and {\tt al2} timed out at \texttt{9} latent terms, and {\tt i9} timed out at \texttt{8} latent terms. Both results are for locked circuits with variables selected with the depth of locked expression, $\hat{e_i}$, equal to \texttt{1}.

\subsection{Impact of Expression Selection on Attack Resilience}\label{sec:attackvsexprsel}
\subsubsection{Attack resilience vs. Depth of locked expression.}\label{sec:attackvsdepth}
The attack resiliency of $\lckd$ increases significantly as we increase the depth of the locked expression selected for HOLL for $\ckt$.
We observe that for a number of latent terms in key relation equal to 2, for benchmark {\tt al2}, increases from 213s to 3788s for depth 1 and 2, respectively. For benchmark {\tt i9}, attack time increases from 351s to 1141s for depth 1 and 2, respectively.

\subsubsection{Attack time vs. Coverage.}\label{sec:attackvscoverage}
To show the effect of coverage we select expressions (in $e_i \in E$) such that the distance (\S\ref{sec:exprsel}) among the expressions is largest (termed as diverse) and smallest (termed as converged).
When the coverage of locked expressions is larger, i.e. locked expressions are spread out in the circuit, the attack requires more time than breaking locked expressions limited in one locality. For example, for benchmarks {\tt C432} and {\tt i9}, attack time increases from 115s to 142s and 229s to 316s, respectively, when expression selection heuristic is changed from converged to diverse.
The results are with three latent terms.

\subsection{Hardware cost}\label{sec:hardwarecosteval}
\begin{table}[t]
\caption{Hardware Impact of HOLL.}\label{impact_of_locking}
\centering
\begin{tabular}{ c | X | X X X | X }
\hline
& \multicolumn{1}{c|}{\bf Orig.} & \multicolumn{3}{c|}{\bf Key Relation} & {\bf Over- head} \\
{\bf Bench} & {\bf Area ($\mu m^2$)} &	{\bf Area ($\mu m^2$)} 	&	{\bf \#Eq. LUT}	& {\bf \#Eq. conf. bits} & {\bf Area (\%)} \\\hline
{\tt al2}	&	17.89	&	4.473	&	138	&	8,832 &	25.0\\
{\tt cht}	&	20.74	&	4.178	&	132	&	8,448 & 20.1	\\
{\tt C432}	&	20.05	&	4.866	&	150	&	9,600 &	24.3\\
{\tt C880}	&	50.04	&	4.325	&	132	&	8,448 &	8.6\\
{\tt i9}	&	77.07	&	4.129	&	126	&	8,064 &	5.4\\
{\tt i7}	&	80.41	&	4.129	&	126	&	8,064 &	4.3\\
{\tt x3}	&	95.21	&	5.014	&	156	&	9,984 &	5.3\\
{\tt frg2}	&	100.81	&	4.669	&	144	&	9,216 &	4.6	\\
{\tt i8}	&	120.37	&	4.325	&	132	&	8,448 &	3.6\\
{\tt des}	&	445.37	&	5.554	&	174	&	11,136 &	1.2\\
\hline
\end{tabular}
\end{table}
The key relations can be implemented either as embedded Field Programmable Gate Array (eFPGA) or Programmable Array Logic.

We synthesize the original and locked designs with {\sc Synopsys Design Compiler} R-\texttt{2020}.\texttt{09}-SP1 targeting the Nangate \texttt{15}nm ASIC technology at standard operating conditions (\texttt{25}$^{\circ}$C).

Table~\ref{impact_of_locking} provides the estimated cost for implementing the key relations with programmable devices. To do so, we compute the number of equivalent NAND2 gates used to estimate the number of \texttt{6}-input LUTs. Given the number of LUTs, we provide an estimation of the equivalent number of configuration bits--including those for switch elements.
Results show that the size of the key relations is independent of the size of the original design.
Table \ref{impact_of_locking} reports the fraction of the area locked with {\toolname} (\textit{key relation}) to the area of the {\em original} circuit. The results show that the impact of HOLL is low, especially for large designs.

\subsection{HOLL Lock Inference Time}\label{sec:inferencetime}
\figref{synthtime} shows the size of the circuit has minimal impact on the time for inferring the key relations, thanks to backslicing (\S\ref{sec:backslicingexprsel}). The number of gates in the benchmarks increases as we move along the x-axis, which is on a logarithmic scale. The data is for \texttt{100} benchmarks with the number of latent terms in the range $\texttt{[8,10]}$. \figref{keybittime} shows the number of relation bits has a linear impact on the time for key relation inference. We select expression selection depth in range $\texttt{[2-4]}$. HOLL timeouts (\texttt{20} minutes) for inferring the key relation when lower limit of depth is increased to \texttt{3} or upper limit to \texttt{5} for benchmark {\tt al2}. We try to keep the depth as high as possible so that the attack resilience for SynthAttack is more (\S\ref{sec:attackvsexprsel})

\begin{figure}[t]
\centering
\begin{subfigure}[b]{0.45\textwidth}
\includegraphics[width=1\textwidth]{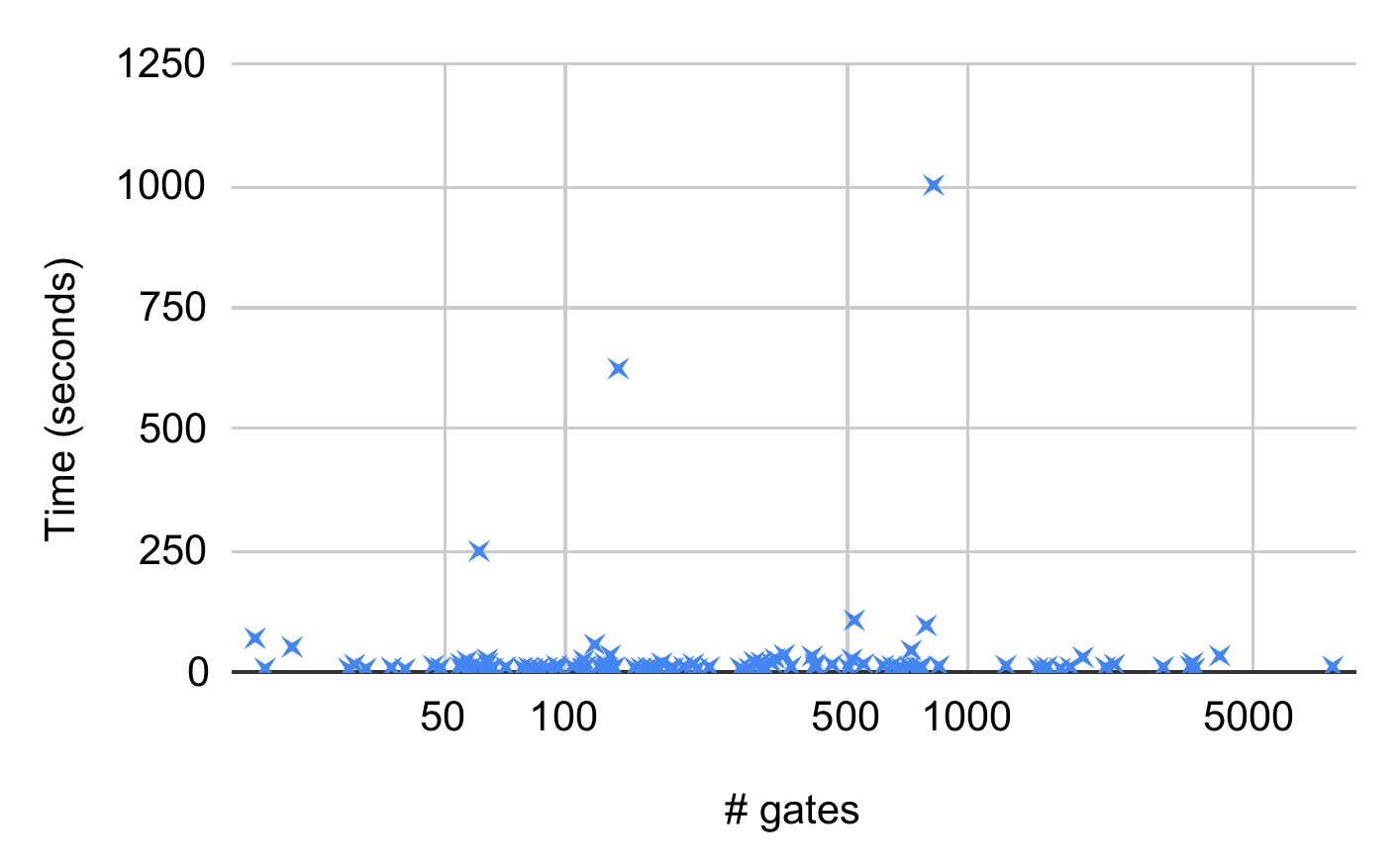}
\caption{Time for 100 benchmarks}\label{fig:synthtime}
\end{subfigure}
\begin{subfigure}[b]{0.45\textwidth}
\includegraphics[width=1\textwidth]{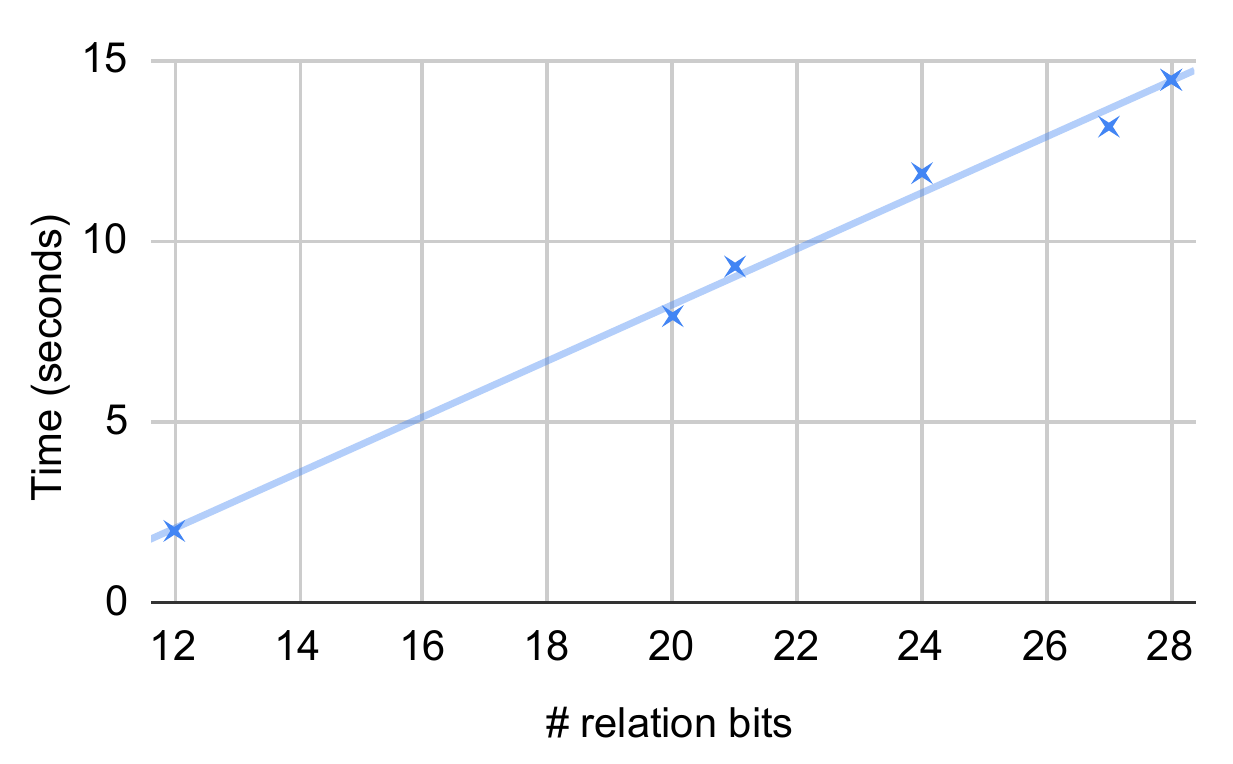}
\caption{Time vs \# relation bits for {\tt al2}}\label{fig:keybittime}
\end{subfigure}
\caption{Execution time for program synthesis.}\label{fig:time}
\end{figure}

\subsection{Impact of HOLL Optimizations}\label{sec:impactoptimization}
\subsubsection{Incremental vs. Monolithic Synthesis}

This experiment shows the advantage of incremental synthesis versus monolithic synthesis (\S\ref{sec:incrementalsynt}).
As can be seen in \tabref{incrementalvsmono} for benchmark {\tt x3}, monolithic synthesis does not scale beyond two expressions (even with a \texttt{1} day timeout for synthesis) while incremental synthesis synthesizes key relations for \texttt{7} locked expressions within \texttt{15}s. We provided exactly the same set of expressions in both cases. While the results of incremental synthesis may be sub-optimal, we found them close to the optima in most cases. \figref{keyrelationx3} shows the key relations synthesized by the two settings when the threshold for the number of expressions is \texttt{2}: both the algorithms generate compact relations with \texttt{8} terms (five stimulus and three latent terms) and using \texttt{8} relation bits.


\begin{table}[t]
\caption{Impact on HOLL time.}
\centering
\begin{subtable}{.4\textwidth}
\centering
\caption{Incremental for {\tt x3}} \label{tab:incrementalvsmono}
\begin{tabular}{| c | c | c | c | c | c | c | c | c |}
\hline
\centering
 \multirow{2}{*}{\bf Time (sec.)}& \multicolumn{7}{|c|}{\bf \# of expressions}\\\cline{2-8}
&{\bf 1} &{\bf 2} &{\bf 3} &{\bf 4} &{\bf 5} &{\bf 6} &{\bf 7}\\\hline\hline
Monolithic &1 &31 &\multicolumn{5}{|c|}{timeout 4 hrs} \\\hline
Incremental &1 &2 &8 &9 &12 &13 &14 \\\hline
\end{tabular}
\end{subtable}
\hspace{5mm}
\begin{subtable}{0.55\textwidth}
\centering
\caption{Impact of backslicing (bkslice) for {\tt al2}} \label{tab:backslicing}
\begin{tabular}{|l | l | l | l | l | l | l | l | l |}
\hline
\centering
\multirow{2}{*}{\bf Time (sec.)}& \multicolumn{7}{|c|}{\bf \# of expressions}\\\cline{2-8}
& {\bf 1} &{\bf 2} &{\bf 3} &{\bf 4} &{\bf 5} &{\bf 6} &{\bf 7}\\\hline
 with bkslice &4 &\multicolumn{6}{|c|}{timeout 4 hrs} \\\hline
 w/o bkslice &1 &225 &512 &513 &514 &548 &642 \\\hline
\end{tabular}
\end{subtable}
\end{table}

\begin{figure}[!]
\centering
\begin{subfigure}[b]{0.45\textwidth}
\{$(r_0 \bimp i0), (r_1 \bimp j0),$ \\
$(r_2 \bimp d2), (r_3 \bimp m0),$ \\
$(r_4 \bimp l0), (r_5 \bimp (r_1 \lor r_0)),$ \\
$(r_6 \bimp (\lnot r_3)), (\\
r_7 \bimp (r_2 \land r_4))$\}
\caption{Monolithic synthesis}\label{fig:monolithicx3}
\end{subfigure}%
\begin{subfigure}[b]{0.45\textwidth}
\{$(r_1 \bimp i0), (r_2 \bimp j0),$ \\
$(r_3 \bimp d2), (r_0 \bimp m0),$ \\
$(r_4 \bimp l0), (r_6 \bimp (r_2 \lor r_1)),$ \\
$(r_5 \bimp (\lnot r_0)), \\
(r_7 \bimp (r_3 \land r_4))$\}
\caption{Incremental synthesis}\label{fig:incrementalx3}
\end{subfigure}
\caption{Key relations generated for two expressions of {\tt x3}.}\label{fig:keyrelationx3}
\end{figure}

\subsubsection{Backslicing}
 Backslacing (\S\ref{sec:backslicingexprsel}) is a critical optimization as shown in \tabref{backslicing}. For example, if backslicing is disabled for the benchmark {\tt al2}, HOLL is not able to generate a key relation beyond \texttt{2} expressions with a timeout of \texttt{4} hours. Instead, when backslicing is active, the HOLL performance is linearly dependent on the number of locked expressions to infer. These results are with depth of variable selection in range \texttt{[2, 4]}.

%% file: content/related_work.tex
\section{Related Work}\label{sec:relatedwork}
\subsubsection{Logic Locking: Attacks and Defenses.}
Existing logic locking methods aptly operate on the gate-level netlists~\cite{llcarxiv}. Gate-level locking cannot obfuscate all the semantic information because logic synthesis and optimizations absorb many of them into the netlist before the locking step. For example, constant propagation absorbs the constants into the netlist. Recently, alternative high-level locking methods obfuscate the semantic information before logic optimizations embed them into the netlist~\cite{dac18,rtlOzgur}. For example, TAO applies obfuscations during HLS~\cite{dac18} but requires access to the HLS source code to integrate the obfuscations and cannot obfuscate existing IPs. Protecting a design at the register-transfer level (RTL) is an interesting compromise~\cite{obfuscateDSPRTL,5401214}. Most of the semantic information (e.g., constants, operations, and control flows) is still present in the RTL and obfuscations can be applied to existing RTL IPs. In~\cite{obfuscateDSPRTL}, the authors propose structural and functional obfuscation for DSP circuits. In~\cite{5401214}, the authors propose a method to insert a special finite state machine to control the transition between obfuscated mode (incorrect function) and normal mode (correct function). Such transitions can only happen with a specific input sequence. Differently from~\cite{9218519}, we extract the relation directly from the analysis of a single RTL design, making the approach independent of the design flow. None of these methods consider the possibility of hiding a relation among the key bits.

\subsubsection{Program Synthesis.}

 Program synthesis has been successful in many domains: synthesis of heap manipulations~\cite{Roy:2013,Garg:2015,Verma:2017}, bit-manipulating programs~\cite{Jha:2010}, bug synthesis~\cite{BugSynthsis}, parser synthesis~\cite{Leung:2015,Singal:2018}, regression-free repairs~\cite{Bavishi:2016, Bavishi:2016b}, synchronization in concurrent programs~\cite{Verma:2020}, boolean functions~\cite{Manthan1,Manthan2,ManthanDQBF} and even differentially private mechanisms~\cite{Kolahal}.
There has also been an interest in using program synthesis in hardware designs~\cite{cook2009}. VeriSketch~\cite{ardeshiricham2019} exploits the power of program synthesis in hardware design. Our work is orthogonal to the objectives and techniques of VeriSketch: while  VeriSketch secures hardware against timing attacks, we propose a hardware locking mechanism. Zhang et al.~\cite{zhang2020} use SyGUS based program synthesis to infer environmental invariants for verifying hardware circuits. We believe that this work shows the potential of applying programming languages techniques in hardware design. We believe that there is also a potential of applying program analysis techniques, both symbolic~\cite{KLEE,DART,Hydra,Khurana:2017,Pandey:2019}, dynamic~\cite{kiteration,chouhan2013} and statistical~\cite{tarantula:2002,liblit:2005,Ulysis,Modi:2013}, for hardware analysis; this is a direction we intend to pursue in the future.


%% file: content/discussion.tex
\section{Discussion}\label{sec:Discussion}

We end the paper with an important clarification: the \textit{eFPGA configuration} in HOLL can also be represented as a bit sequence (i.e., a sequence of \textit{configuration} bits). So, why can an attacker not launch attacks similar to SAT attacks on logic locking to recover the HOLL configuration bitstream?

The foremost reason is that while the key-bits in traditional logic locking simply represent a \textit{value} that the attacker attempts to recover, the bit-sequence in HOLL is an encoding of a \textit{program}~\cite{clift2020, petersen2019}. This raw bit-sequence used to program an eFPGA is too ``low-level" to be synthesized directly---the size of such bit-streams is about \textbf{60-85} times of the keys used in traditional logic locking (128 key bit-sequence). So, the HOLL algorithm designer uses a higher-level domain-specific language (DSL) to synthesize the key relation (see \S\ref{sec:dsl}), that is later ``compiled" to the configuration sequence. The attacker will also have to use a similar strategy of using a high-level DSL to break HOLL.

However, while the designer of the key relation can use a well-designed small domain-specific language (DSL) that includes the exact set of components required (and a controlled budget) to synthesize the key relation, the attacker, not aware of the key relation or the DSL, will have to launch the attack with a ``guess" of a large overapproximation. In other words, \textbf{the domain-specific language used for synthesis is also a \textit{secret}}, thereby making HOLL much harder to crack than traditional logic locking. 

We evaluate HOLL (\S\ref{sec:attackreselience}) under the assumption that the DSL (and budget) are known to the attacker. In real deployments (when the DSL is not known to the attacker), HOLL will be still more difficult to crack.